\documentclass[11pt]{article}
\usepackage{graphicx}
\usepackage{a4wide}
\usepackage{amsmath,amsthm}
\usepackage{amssymb}
\usepackage{tikz}
\usepackage{pgfplots}
\usetikzlibrary{decorations.pathmorphing}
\usetikzlibrary{positioning}
\tikzstyle{sommet}=[circle,draw,fill=black,inner sep=2pt]
\tikzstyle{opt}=[draw=red,dashed]
\tikzstyle{alg}=[draw=blue]
\tikzstyle{algopt}=[double]
\tikzstyle{none}=[dotted]
\tikzstyle{Snake}=[decorate,decoration={snake}]
\newcommand{\etal}{\emph{et al}.\ }

\pgfplotsset{compat=1.10}

\newtheorem{theorem}{Theorem}

\newtheorem{lemma}{Lemma}
\newtheorem{proposition}[lemma]{Proposition}

\newtheorem{Remarks*}{Remarks}

\begin{document}

\title{Online Maximum Matching with Recourse\thanks{Supported by ANR OATA, DIM RFSI DACM and Labex Mathématique Hadamard. Preliminary version appeared in the Proceedings of the 43rd International Symposium on Mathematical Foundations of Computer Science (MFCS), 2018.}}

\author{Spyros Angelopoulos\thanks{Sorbonne Universit\'e, CNRS, Laboratoire d'informatique de Paris 6, LIP6, F-75252 Paris, France.} \and Christoph D\"urr\footnotemark[2] \and Shendan Jin\footnotemark[2]}

\maketitle

\begin{abstract}
We study the online maximum matching problem in a model in which the edges are
associated with a known recourse parameter $k$. An online algorithm for this
problem has to maintain a valid matching while edges of the underlying graph
are presented one after the other. At any moment the algorithm can decide to
include an edge into the matching or to exclude it, under the restriction that
at most $k$ such actions per edge take place, where $k$ is typically a small
constant. This problem was introduced and studied in the context of general
online packing problems with recourse by Avitabile \etal [Information Processing Letters, 2013], 
whereas the special case $k=2$ was studied by Boyar \etal [WADS 2017].

In the first part of this paper we consider the {\em edge arrival} model, in
which an arriving edge never disappears from the graph. Here, we first show an
improved analysis on the performance of the algorithm AMP of Avitabile {\em et al.}, by exploiting the structure of the
matching problem. In addition, we show that the greedy algorithm has
competitive ratio $3/2$ for every even $k$ and ratio $2$ for every odd $k$.
Moreover, we present and analyze an improvement of the greedy algorithm which
we call $L$-\textsc{Greedy}, and we show that for small values of $k$ it
outperforms the algorithm AMP. In terms of
lower bounds, we show that no deterministic algorithm better than $1+1/(k-1)$
exists, improving upon the known lower bound of $1+1/k$.

The second part of the paper is devoted to the \emph{edge arrival/departure
model}, which is the fully dynamic variant of online matching with recourse.
The analysis of $L$-\textsc{Greedy} and AMP carry through in this model;
moreover we show a lower bound of $(k^2-3k+6) / (k^2-4k+7)$ for all even $k
\ge 4$. For $k\in\{2,3\}$, the competitive ratio is $3/2$.
\end{abstract}


\section{Introduction}
\label{sec:introduction}

In the standard framework of online computation, the input to the algorithm is
revealed incrementally, i.e., as a sequence of {\em requests}. For each such requested
input item, the online algorithm must make a decision that is typically {\em
irrevocable}, in the sense that the algorithm commits, in a permanent manner,
to the decision associated with the request. More precisely, the algorithm may
not alter any previously made decisions while considering later requests. This
rather stringent constraint is meant to capture what informally can be
described as ``the past cannot be undone''. Equally significantly, it is at
the heart of adversarial arguments that can be used to argue that its performance, measured by means of the
{\em competitive analysis} framework~\cite{textbook-borodin}, cannot be improved beyond a
certain bound.

Nevertheless, there are real-life applications in which some (limited)
rearrangement of the online solution during the execution of the algorithm may
be doable, or even requisite. For instance, online call admission protocols
may sporadically reconfigure the virtual paths assigned in the network. For a
different example, in online scheduling (or resource allocation) problems, it
may be permissible for a job to be transfered to a processor other than the
one specified by the original decision associated with the job.  Clearly, a
trade-off is to be found between the guaranteed competitive ratio and the cost
of re-optimizing the current solution.  Different approaches to this objective have been
considered.  One such approach has studied the minimum total re-optimization
cost required in order to maintain an optimal solution, see Bernstein \etal \cite{Bernstein:2018aa}.
Another approach has focused on the best achievable competitive ratio when there is some bound on the allowed
re-optimization, such as in Avitabile \etal
\cite{AvitabileMathieu:13:Online}, and it is the main model we consider in this work.

More specifically, we study the \emph{online maximum matching}
problem, in which the objective is to maintain a vertex-disjoint edge set of maximum
cardinality in a given graph. Here, the request sequence consists of the edges of the graph, and for each revealed edge, the 
online algorithm must decide whether to include it in the matching or not. In the standard model, each request has to be
immediately either accepted in the matching, or rejected, and this decision cannot be revoked in the future. We consider a generalization 
of this model, in which the algorithm can switch
between accepting and rejecting an edge that has already appeared, but is
allowed up to $k$ such modifications per edge, where $k$ is the {\em recourse parameter} that is known to the algorithm\footnote{For consistency with other recourse models, 
we define the initial default state of an edge as \emph{rejected}, and therefore rejecting a newly arriving edge does not count as decision modification. See also Section~\ref{subsec:related.work.models}.}. The following section provides the formal definition of the problem, as well as a motivating application.

\subsection{Problem definition and motivation}
\label{subsec:definition}

We define the {\em online matching problem with $k$ edge-recourse}, for given $k \in \mathbb N^*$. The request sequence is a permutation of the edge set 
$E$ in a graph $G=(V,E)$. The online algorithm knows the set $V$, as well as the parameter $k$, but not the set $E$. At each point in time, the online algorithm must maintain a matching
$M\subseteq E'$, where $E' \subseteq E$ is the set of currently revealed edges. Specifically, for each revealed edge $e$ in the sequence, the online algorithm either {\em accepts} $e$, by adding it in its
matching, or {\em rejects} it. In addition, the algorithm must obey the 
edge recourse constraint, which is defined as follows.  Every
edge has an integer type, which is set to $0$ upon the arrival of an edge. Whenever the
algorithm decides either to include an edge $e$ to its matching or to remove it,
the type of $e$ is increased.  The algorithm can perform these operations at any point in time
subject to the constraint that no edge type exceeds $k$.

The objective of the problem is to design an online algorithm of minimum {\em competitive ratio} which is defined as the worst-case ratio, over all
request sequences $\sigma$, of the cardinality of an optimal offline matching to the cardinality of the matching produced by the online algorithm.
See Section~\ref{subsec:competitive.analysis} for a more elaborate discussion on this measure. 

We make a distinction between two settings concerning the revealed edges. In the first setting, which we call the {\em edge arrival} setting, once an edge appears 
(as a request) it is guaranteed to be part of the input graph $G$. In other words edges may only arrive, but never depart. A more general setting is the one in which edges may not only arrive (in the form of a request), but may also disappear adversarially
(subsequently to their appearance). More specifically, at each request, the online algorithm is also informed about the subset of edges that are no longer part of the input graph since the last request.
We call this setting the {\em edge arrival/departure} setting. Note that in this model, a given edge $e$ may appear and disappear several times, and the number of such events is unrelated to the recourse parameter $k$.
This setting is motivated by similar models that have been studied in the context of the online Steiner tree
problem~\cite{GuptaKumar:14:Deletions}, and models the fully dynamic variant of the problem at hand. 

Last, in what concerns the edge arrival/departure setting, we make a further distinction 
concerning the edge departures. In the {\em full departure} model, the
adversary is allowed to delete any edge in the graph, and thus also any edge
that may have been provisionally accepted by the online algorithm. We show that this model
is quite restrictive, since it yields excessive power to the
adversary. We thus also study the {\em limited departure} model, in which the
adversary may delete only edges not currently accepted by the online
algorithm.

To motivate the problem and the various models, consider the following application related to resource allocation. 
Suppose that we have a set of tasks $T$ and a set of workers $W$, as well as a bipartite graph $G=(W,T,E)$.
An edge $(w,t)\in E$ between a worker $w$ and a task $t$ signifies that $w$ is qualified to work on $t$, and we assume that a task can be assigned only to one worker and, vice versa, a worker can only be occupied in one task. A maximum matching in $G$ describes a maximum assignment of tasks to workers.  In the standard online version of the problem
(with no recourse), the input consists of pairs of the form $(w,t)$, and once an algorithm assigns worker 
$w$ to task $t$ (namely, accepts the edge $(w,t)$) it cannot assign $w$ to a task other than $t$ or $t$ to a worker other than $w$.
In contrast, the online matching problem with $k$ edge-recourse captures the application in which 
$w$ may be assigned in an ``on/off'' manner to task $t$ up to $k$ times. Here $k$ bounds the willingness of workers to be reassigned to tasks.

Last, the edge arrival/departure model can capture the dynamic situation in which the compatibility of tasks and worker skills can change
over time. Under the full departure model, there are no constraints on such changes. In contrast, the
limited departure model stipulates that if worker $w$ remains assigned to task $t$ by the online algorithm, then $w$ does not lose his
qualification for $t$, in the sense that the worker has a continual occupation
with the said task and maintains the required skills for the task. However,
once the online algorithm decides to remove worker $w$ from task $t$ (i.e.,
the online algorithm provisionally rejects edge $(w,t)$), then the worker
might lose its qualification for the task over time.

\subsection{Competitive analysis}
\label{subsec:competitive.analysis}

Given an online algorithm ALG, and a request sequence $\sigma$, we use the standard
notation $\textrm{ALG}(\sigma)$ to denote the output of ALG on sequence $\sigma$.
In the context of the matching problem, we will denote by $|\textrm{ALG}(\sigma)|$
the {\em value} of the matching output by ALG on sequence $\sigma$, namely the 
cardinality of its matching. We will denote by OPT the offline optimal algorithm 
that has knowledge of the request sequence, and hence $|\textrm{OPT}(\sigma)|$
denotes the cardinality of the offline optimal matching. 

An online algorithm $\textrm{ALG}$ for the maximum cardinality matching problem is said to be
\emph{asymptotically $c$-competitive} if there is a constant $d$ such that
$|\textrm{ALG}(\sigma)| \geq |\textrm{OPT}(\sigma)| / c - d$, for all request
sequences $\sigma$. If $d=0$
the algorithm is called \emph{strictly $c$-competitive}.  Note that some
previous work on the matching problem has used the reciprocal ratio. The
smallest $c$ for which an online algorithm $\textrm{ALG}$ is $c$-competitive
is called the \emph{asymptotic competitive ratio} of $\textrm{ALG}$.  The
strict competitive ratio is defined similarly.  If it so happens that this
minimum value does not exist, the competitive ratio is actually defined by the
corresponding infimum. In this setting an \emph{upper bound} on the (strict or
not) competitive ratio establishes the performance guarantee of an online
algorithm, whereas a \emph{lower bound} is a negative result.

Both upper and lower bounds in this work are shown for the strict competitive
ratio. This implies that the upper bounds carry over to the more general
definition, but this generalization does not necessarily hold for the lower
bounds. We emphasize, however, that the known lower bounds for edge-bounded
recourse problems in~\cite{AvitabileMathieu:13:Online,BoyarFavrholdt:17:Relaxing-the-Irrevocability} 
are likewise expressed in terms of the strict
competitive ratio. This is due, perhaps, to difficulties in applying
techniques that extend the lower bounds to the standard definition of the
competitive ratio that are inherent to the recourse setting, and which do not
arise in the traditional online framework of irrevocable decisions.
Specifically, it is not obvious how to use techniques based on multiple copies
of an adversarial instance in order to lower-bound the performance of any
online alghorithm, although this may be possible for specific online
algorithms (see, e.g., Lemma~\ref{lem:lb_asym_LGREEDY}). For convenience, we
will henceforth refer to the strict competitive ratio as simply the
``competitive ratio''.

For convenience of notation, we will omit the request sequence $\sigma$ when it is implied
from context, or when it is not relevant. Thus, with a slight, but standard abuse of notation, 
we will denote by ALG both the algorithm and its output. 

\subsection{Related work}
\label{subsec:related.work}

Several online combinatorial optimization problems have been studied under recourse
settings. The broad objective is to quantify the trade-off between the
competitive ratio and a measure on the modifications allowed on the solution.
Some representative examples include online problems such as minimum spanning
trees and TSP~\cite{MegowSkutella:16:The-power}, Steiner
trees~\cite{GuptaKumar:14:Deletions,GuGupta:16:The-power}, knapsack
problems~\cite{Han:2009aa,IwamaTaketomi:02:Removable}, assignment
problems in bipartite unweighted graphs~\cite{GuptaKumar:14:Maintaining}, and
general packing problems~\cite{AvitabileMathieu:13:Online}. In the remainder
of this section we review work related to online maximum matching.

\subsubsection{Online matching, with and without recourse: models}
\label{subsec:related.work.models}

Concerning online matching, two different request models have been studied in the past. In the \emph{vertex arrival model}, vertices arrive in online
fashion, revealing, at the same time, the edges incident with previously arrived
vertices. This model has mainly been considered for bipartite graphs, with
left side vertices arriving online, and right side vertices being initially
known (see the survey~\cite{mehta2013online}).  In the \emph{edge arrival model} the edges arrive online in arbitrary order, revealing at the same time incident vertices\footnote{We emphasize
that in our work we consider the maximum cardinality matching problem; some
previous work (with or without recourse) has considered the generalized {\em
weighted matching} problem, in which each edge has a weight and the objective
is to maximize the weight of edges in the matching.}.

In the standard online model, every request (either vertex, or edge, depending on the request model) is served in an irrevocable manner. In contrast, for online matching with recourse and edge arrivals, several models have been proposed that relax the irrevocable nature of a decision.
In the {\em late reject} model
\cite{BoyarFavrholdt:17:Relaxing-the-Irrevocability}, which is also called the
{\em preemptive model} \cite{Chiplunkar:2015aa}, an edge can be accepted only upon its
arrival, but can be later rejected. The problem we study in this work, namely online matching with $k$ edge-recourse was 
introduced in~\cite{AvitabileMathieu:13:Online}.
Boyar \etal \cite{BoyarFavrholdt:17:Relaxing-the-Irrevocability}
refer to this model for $k=1$ as the {\em late accept} model, and for $k=2$ as
the {\em late accept/reject} model.
Clearly, the competitive ratio is monotone in $k$.  
Figure~\ref{fig:models} provides an illustration of the algorithm's actions under the different models.

\begin{figure}[ht]
\begin{tabular}{l}
    \begin{tikzpicture}
            \node (release) at (0,0) [draw,fill=black!20] {release};
            \node (dummy)   [right=of release] {};
            \node (reject0) at (2.5,0.5) [draw] {reject};
            \node (accept1) at (2.5,-0.5) [draw] {accept};
            \draw [->] (release) -- (reject0);
            \draw [->] (release) -- (accept1);
    \end{tikzpicture}
\\
Standard model
\\[1em]
 \begin{tikzpicture}
            \node (release) [draw,fill=black!20] {release};
            \node (accept0) [draw,right=of release] {accept};
            \node (reject1) [draw,right=of accept0] {reject};
            \draw [->] (release) -- (accept0);
            \draw [->] (accept0) -- (reject1);
            \draw [->] (release) to [bend left=25] (reject1);
    \end{tikzpicture}
\\
Late reject model, also called free disposal or preemptive model
\\[1em]
 \begin{tikzpicture}
            \node (release) [draw,fill=black!20] {release};
            \node (reject0) [draw,right=of release] {reject};
            \node (accept1) [draw,right=of reject0] {accept};
            \draw [->] (release) -- (reject0);
            \draw [->] (reject0) -- (accept1);
            \draw [->] (release) to [bend left=25] (accept1);
    \end{tikzpicture}
\\
Late accept model, also called  1 edge-recourse model
\\[1em]
    \begin{tikzpicture}
            \node (release) [draw,fill=black!20] {release};
            \node (reject0) [draw,right=of release] {reject};
            \node (accept1) [draw,right=of reject0] {accept};
            \node (reject2) [draw,right=of accept1] {reject};
            \draw [->] (release) -- (reject0);
            \draw [->] (reject0) -- (accept1);
            \draw [->] (accept1) -- (reject2);
            \draw [->] (release) to [bend left=25] (accept1);
    \end{tikzpicture}
\\
Late accept/reject model, also called edge 2 edge-recourse model
\\[1em]
    \begin{tikzpicture}
            \node (release) [draw,fill=black!20] {release};
            \node (reject0) [draw,right=of release] {reject};
            \node (accept1) [draw,right=of reject0] {accept};
            \node (reject2) [draw,right=of accept1] {reject};
            \node (accept3) [draw,right=of reject2] {accept};
            \node (reject4) [draw,right=of accept3] {reject};
            \node (a) [below=of accept1,yshift=8mm] {};
            \node (b) [below=of reject4,yshift=8mm] {};
            \draw [<->,dotted] (a) -- (b) node [midway] {$k$};
            \draw [->] (release) -- (reject0);
            \draw [->] (reject0) -- (accept1);
            \draw [->] (accept1) -- (reject2);
            \draw [->] (reject2) -- (accept3);
            \draw [->] (accept3) -- (reject4);
            \draw [->] (release) to [bend left=25] (accept1);
    \end{tikzpicture}
\\
$k$ edge-recourse model
\end{tabular}
\caption{Illustration of the actions of an online matching algorithm under the different edge-arrival models with recourse. }
\label{fig:models}
\end{figure}

\subsubsection{Online matching, with and without recourse: known results}

For online maximum matching without recourse, and for the vertex arrival model, 
the seminal work of Karp {\em et al.}~\cite{KarpVazirani:90:An-optimal} gave a randomized online algorithm with
competitive ratio $e/(e-1)$ in the vertex arrival model together with a matching
lower bound on any online algorithm. In contrast, for the edge arrival model 
and the randomized competitive ratio, \cite{BuchbinderSegev:17:Online} showed
a lower bound of $(3+1/\varphi^2)/2$ as well as an upper bound of $1.8$
for the special case of forests, where $\varphi$ is the golden ratio.

It is well known that any inclusion-wise maximal matching has cardinality
at least half of the optimal maximum cardinality matching.  From this it
follows that the greedy online algorithm, which accepts an edge as long as it can be added to the current matching,  
has competitive ratio at most $2$, which in the standard model is optimal among all deterministic online
algorithms.

\paragraph{Late reject} \
In the vertex arrival model, the greedy algorithm achieves
trivially the competitive ratio of $2$, which is optimal for all
deterministic online algorithms.
The situation differs in the edge arrival model.  Epstein {\em et
al.}~\cite{EpsteinLevin:13:Improved} showed that for online weighted matching, the deterministic competitive ratio is exactly $3+2\sqrt2 \approx
5.828$, as the upper bound of \cite{McGregor:05:Finding} matches the lower
bound of \cite{Varadaraja:11:Buyback}.  The same paper
\cite{EpsteinLevin:13:Improved} shows that the randomized competitive ratio is
between $1+\ln2\approx 1.693$ and $5.356$. Chiplunkar {\em et al.}~\cite{Chiplunkar:2015aa} presented a randomized $28/15$-competitive algorithm based on a primal-dual analysis.

\paragraph{$k$ edge-recourse} \
This model was introduced and studied by Avitabile {\em et
al.}~\cite{AvitabileMathieu:13:Online} for the edge arrival setting, in the
context of a much broader class of online packing problems. They gave an
algorithm, which we call AMP, that combines doubling techniques with optimal
solutions to offline instances of the problem, and which has  competitive ratio $1 + O \left( \frac{\log k}{k} \right)$
(see Section~\ref{sec:upper_bound} for an analysis of AMP). We note that this result is formulated in~\cite{AvitabileMathieu:13:Online} in a ``dual'' setting. More precisely,~\cite{AvitabileMathieu:13:Online} asks the question:
how big should the edge budget $k$ be such that there is a $(1+\varepsilon)$-competitive online
algorithm that makes at most $k$ changes per edge? They showed that
$k=O(\ln(1/\varepsilon)/\varepsilon)$ suffices.
On the negative side, they showed that no randomized algorithm can be better than
$1+1/(9k-1)$-competitive; we note also that their construction implies a
lower bound of $1+1/k$ for all deterministic algorithms.

Boyar \emph{et al.}~\cite{BoyarFavrholdt:17:Relaxing-the-Irrevocability}
showed that the deterministic competitive ratio is $2$, if
$k=1$, and $3/2$, if $k=2$, and this optimal competitive ratio is achieved by the greedy algorithm.
Moreover~\cite{BoyarFavrholdt:17:Relaxing-the-Irrevocability}
studied several other problems for a value of the recourse parameter equal to 2, such as
independent set, vertex cover and minimum spanning forest.

\paragraph{Minimizing recourse} \
Bernstein {\em et al.}~\cite{Bernstein:2018aa} studied a different recourse model
in which the algorithm has to maintain an optimal matching.
More specifically, recourse here is expressed in terms of 
the total number of times edges enter or leave the algorithm's matching. They
considered the setting of a bipartite graph and the vertex arrival model and
showed that a simple greedy algorithm achieves optimality using $O(n \log^2 n)$ replacements, where
$n$ is the number of nodes in the arriving bipartition, whereas the
corresponding lower bound for any replacement strategy is $\Omega(n \log n)$.

\subsection{Contribution of this work}
\label{subsec:contribution}
In the first part of this work, we study the online matching problem with edge $k$-bounded recourse under the edge arrival model. For this problem, we
provide improvements on both upper and lower bounds on the competitive ratio. First, we revisit the doubling algorithm of~\cite{AvitabileMathieu:13:Online} that was originally
analyzed in the general context of online packing problems. We give a better analysis, specifically for the problem at hand, that uses concepts and ideas related to the matching problem; we also show that the AMP algorithm has  competitive ratio $1+O(\frac{\log k}{k})$. On the negative side, we show that
no deterministic algorithm is better than $(1+1/(k-1))$-competitive, improving upon the known bound of $1+1/k$ of~\cite{AvitabileMathieu:13:Online}.

At first sight these improvements may seem marginal; however one should take
into consideration that $k$ is typically a small parameter, and thus the
improvements are by no means negligible. In this spirit, we propose and
analyze a variant of the greedy algorithm which we call $L$-Greedy. This
algorithm applies, at any step, augmenting paths as long as their length is at
most $2L+1$. We show that for a suitable choice of $L$, this algorithm is
$1+O(1/\sqrt{k})$-competitive. While this algorithm is thus not superior to
AMP for large $k$ (and more specifically, to its improved analysis in the
context of the matching problem), for small $k$ (and in particular, for $k
\leq 20$) it does achieve an improved competitive ratio, see
Figure~\ref{fig:compare_ratios}. Moreover, we extend a result of Boyar {\em et
al.}~\cite{BoyarFavrholdt:17:Relaxing-the-Irrevocability} that showed that the
greedy algorithm is $\frac{3}{2}$-competitive for $k=2$ to all even $k$ (for odd $k$,
the competitive ratio is 2).

\begin{figure}[ht]
\begin{center}
\begin{tikzpicture}
\begin{axis}[
  xtick=data,
  xlabel=edge budget k,
  ylabel=competitive ratio]
  \addplot[color=red,mark=*] coordinates {
        (4, 2.64526)
        (6, 2.03971)
        (8, 1.7763)
        (10, 1.62664)
        (12, 1.52919)
        (14, 1.46023)
        (16, 1.40862)
        (18, 1.3684)
        (20, 1.33609)
        (22, 1.3095)   };
\addlegendentry{original AMP analysis}
  \addplot[color=orange,mark=x] coordinates {
(4,2.59808)
(6,1.86919)
(8,1.6136)
(10,1.48058)
(12,1.39808)
(14,1.3415)
(16,1.30008)
(18,1.26833)
(20,1.24315)
(22,1.22264)};
\addlegendentry{improved AMP analysis}
  \addplot[color=blue,mark=o] coordinates {
        (4, 1.5)
        (6, 1.466)
        (8, 1.428)
        (10, 1.3333)
        (12, 1.318)
        (14, 1.307)
        (16, 1.3)
        (18, 1.247)
        (20, 1.2421)
        (22, 1.238)};
\addlegendentry{$L$-\textsc{Greedy}}
\end{axis}
\end{tikzpicture}
\end{center}
\caption{Comparison of the  competitive ratios of the algorithm AMP and the algorithm $L$-\textsc{Greedy}}
\label{fig:compare_ratios}
\end{figure}

In terms of techniques, we analyze both AMP and $L$-Greedy using amortization arguments in which the profit of the algorithms is expressed in terms
of weights appropriately distributed over nodes in the graph. We achieve these improvements by exploiting properties of augmenting paths in matching algorithms.

The second part of the paper is devoted to the \emph{edge arrival/departure
model}, which is the fully dynamic variant of the online matching problem.
First, we observe that the analysis of $L$-Greedy and AMP carries through in
this model as well. On the negative side, we show a lower bound of
$(k^2-3k+6)/(k^2-4k+7)$ for all even $k \ge 4$. For $k\in \{2,3\}$, the
competitive ratio is $3/2$.  We obtain these lower bounds by modeling the
game between the algorithm and the adversary as a game played over
\emph{strings} of numbers 0 up to $k$.  These strings represent alternating
paths, and each number represents how many times the algorithm has modified its
decision on the corresponding edge.  This
provides a simpler combinatorial aspect to the game played between the
adversary and the algorithm.

We note that, for the analysis of AMP and of L-\textsc{Greedy}, we assume that
$k$ is even. This assumption is borrowed from
\cite{AvitabileMathieu:13:Online} and is required for the analysis.  Of course
for odd $k\geq 3$ these algorithms can be run with budget $k-1$, providing
a valid upper bound on the  competitive ratio. Note that our lower bound
in the arrival model holds for all values of $k$.

\subsection{Preliminaries}


A \emph{matching} in a graph $G=(V,E)$ is a set of edges $M\subseteq E$ with
disjoint endpoints.  A vertex $v\in V$ is said to be \emph{matched} by $M$ if
there is an edge $e\in M$ incident to $v$, and is \emph{unmatched} otherwise.
A key concept in maximum matching algorithms is
the notion of an \emph{augmenting alternating path}, or simply
\emph{augmenting path}.  A path $P$ in $G$ is a sequence of vertices
$v_1,v_2,\dots,v_\ell$ for some length $\ell \geq 2$, such that
$(v_i,v_{i+1})\in E$ for all $i=1,\ldots,\ell-1$.  It is said to be
\emph{alternating with respect to $M$} if every other edge of $P$ belongs to
$M$.  Moreover, an alternating path is \emph{augmenting} if the first and the last
vertex in the path is unmatched by $M$. Applying $P$ to $M$ consists in removing from $M$
the edges in $M\cap P$ and adding the edges in $P\setminus M$.  The resulting
matching has cardinality $M+1$, and every previously matched vertex
remains matched.

We define some concepts that will be useful in the analysis of algorithms
throughout the paper. We will associate each edge with a {\em type} which is
an integer in $[0,k]$. An edge is of type $i$ if it has undergone $i$ decision
flips by the algorithm. Hence, for an edge of type $k$, where $k$ is the
recourse budget, its decision has been finalized, and cannot change further;
we call such an edge {\em blocked}. The type of a path $P$ is defined by the
sequence of the types of its edges, and to make this concept unambiguous, we
choose between the two orientations of the path the one that results in the
lexicographically minimal such sequence. Note that when the algorithm applies
some augmenting path $P$ to its current matching $M$, then the type of every
edge in $P$ is increased by $1$. Moreover, the two extreme edges of an
augmenting path are of type 0, because the endpoints of $P$ are unmatched. We
will call a path {\em blocked} if it contains a blocked edge.

Given a request sequence $\sigma$, let $G(\sigma)$ denote the graph induced by $\sigma$. 
For an online algorithm $ALG$, and a connected component $C$ of $G(\sigma)$, we define
the {\em local ratio} in $C$ as the ratio
between the cardinality of the optimum matching and the cardinality of $ALG's$ matching, 
restricted to edges in $C$.

\section{Online matching in the edge arrival model}
\label{sec:arrival}

\subsection{The algorithm AMP}
\label{sec:upper_bound}

We study the performance of an algorithm proposed by Avitabile \etal \cite{AvitabileMathieu:13:Online}
for the more general \emph{online set packing problem}. In this problem sets arrive online, and
the objective is to maintain a collection of disjoint sets that has maximum cardinality.
More specifically,~\cite{AvitabileMathieu:13:Online} proposed a
\emph{doubling algorithm} which is defined for even $k$ only. The algorithm has a
parameter $r>1$ and there is a decision variable for every set which can be
changed at most $k$ times. The algorithm works in phases, sequentially
numbered by an integer $p$.  Initially $p=0$, and the algorithm's current solution is
$\textrm{AMP}_{0}=\emptyset$. Let $\ell$ be the largest integer such that the
optimal solution has value at least $r^{\ell}$, where $\ell$ is defined to be $-\infty$ if
the optimal solution is empty. Whenever this value increases, the algorithm
starts a new phase. We define $\ell(p)$ as the value of $\ell$ during phase $p$, and thus
\begin{equation}
\ell(p) + i \leq \ell(p+i)
\label{eq:monotonicity.phase}
\end{equation}
for every positive integer $i$. At the
beginning of a new phase, all decision variables that have been changed fewer
than $k$ times are set as in OPT, resulting in the current solution
$\textrm{AMP}_{p}$ (note that the algorithm crucially depends on $k$ being
even in order to produce a feasible solution). In between phases, AMP does not make any changes to 
its current solution.

Avitabile \etal show that the value $|AMP|$ of the solution returned by the algorithm is at least
\[
\left(1-\left( \frac{r-1}{r} + \frac{r}{r^k(r-1)}\right) \right) |OPT|,
\]
which implies that
\[
\frac{|OPT|}{|AMP|} \leq \frac{1}{1-\left( \frac{r-1}{r} + \frac{r}{r^k(r-1)}\right)}
                                          = \frac{r^k(r-1)}{r^{k-1}(r-1)-r}.
\]
Thus, for a given $r > 1$, the competitive ratio of AMP is at most
\(
\frac{r^k(r-1)}{r^{k-1}(r-1)-r}.
\)
Let $r_0$ denote the root of $r^{k-1}(r-1)-r = 0$. If $r < r_0$, then for all $k \geq 4$, we have $r^{k-1}(r-1)-r < 0$ since this function is increasing in $r$. We conclude that the competitive ratio of AMP is upper bounded by
\begin{equation}
\label{eq:inf_AMP}
\min_{r>r_0} \frac{r^k(r-1)}{r^{k-1}(r-1)-r}.
\end{equation}

We will show how to obtain an improved analysis of this algorithm in the context of the matching problem.
Since we know optimal algorithms for $k=1,2$ \cite{BoyarFavrholdt:17:Relaxing-the-Irrevocability} and, since we can show optimality of the greedy algorithm for $k=3$ (see Section~\ref{sec:(k-1)/k}), we can assume, for the purposes of the analysis, that $k\geq 4$.
We begin by a restatement of the update phase that will help us exploit the structure of solutions obtained via augmenting paths.
More specifically, on every
edge arrival, the algorithm updates a current optimal solution OPT.  At the beginning of a new phase,
the algorithm  produces
a matching $\textrm{AMP}_p$ obtained from $\textrm{AMP}_{p-1}$ as follows:
every edge $e\in\textrm{AMP}_{p-1}\setminus \textrm{OPT}$ is removed from
the current matching, and every edge $e\in\textrm{OPT} \setminus
\textrm{AMP}_{p-1}$ which is of type strictly smaller than $k$ is added to the
current matching. Note that edges incident with $e$ have been removed, hence $\textrm{AMP}_p$ is indeed a matching.
Also note that all edges added or removed by the algorithm have their type increased by
one.

Since $\textrm{AMP}_{p-1}$ and $\textrm{OPT}$ are matchings, their symmetric
difference, excluding type $k$ edges, consists of alternating cycles and
alternating paths which can be of even or odd length.  This means that the
algorithm simply applies at the beginning of every phase all those alternating
paths and cycles.


Let $\textrm{OPT}_p$ denote the matching produced by OPT as phase $p$ is about to begin. 
From the statement of AMP, we obtain the following series of inequalities, for every phase $p$.
\begin{equation}  
\label{eq:amp_bound_opt}
r^{\ell(p)} \leq |\textrm{OPT}_p|  \leq |\textrm{OPT}| < r^{\ell(p)+1}.
\end{equation}

For any given phase $p\geq 1$, we aim to bound the ratio $|\textrm{OPT}|/|\textrm{AMP}_p|$, since 
this will allow us to bound the competitive ratio of AMP, for a worst-case choice of $p$. 
Note that the type of an edge increases by at most $1$ with each phase. 
Hence, in the beginning of the $k$ first phases AMP
``synchronizes'' with $\textrm{OPT}$ as there are no blocked edges yet, and as a
result during these phases the ratio $|\textrm{OPT}|/|\textrm{AMP}_p|$ does not
exceed $r$, from~\eqref{eq:amp_bound_opt}.

For the remaining phases we need the following argument.
\begin{proposition}
\label{lemma:inf_sol}
For even $k$ and any phase $p \geq k + 1 $, AMP maintains a matching $\textrm{AMP}_p$ of cardinality at least
$r^{\ell(p)} - r^{\ell(p-k+1)+1}$.
\end{proposition}
\begin{proof}
We denote by the \emph{type of a vertex} $v$ the maximum type of the edges incident with $v$ and  by $n_{i,p}$ the number of vertices of type $i$ in phase $p$.
With every phase change the type of a vertex increases at most by $1$.  Hence every vertex of type $k$ in phase $p$ had positive type in phase $p-k+1$. Thus
\[
n_{k,p} \le  \sum_{i=1}^{k} n_{i,p-k+1} \leq 2 \cdot |\textrm{OPT}_{p-k+1}|,
\]
where the last inequality uses the fact that the left hand side counts the number of vertices matched by the algorithm.
In phase $p$, the difference between the cardinality of the optimal matching and the cardinality of AMP's matching 
is at most the number of blocked augmenting paths, and each of them contains at least two type $k$ vertices. Hence,
\begin{align*}
|AMP_{p}| &\geq |\textrm{OPT}_{p}| - \frac12 \cdot n_{k,p}
\\
&\geq |\textrm{OPT}_{p}| - |\textrm{OPT}_{p-k+1}|
\\
            & > r^{\ell(p)} - r^{\ell(p - k +1) + 1},
\end{align*}
where the last inequality follows from~\eqref{eq:amp_bound_opt}.
\end{proof}

Combining Proposition~\ref{lemma:inf_sol} with the bounds \eqref{eq:amp_bound_opt} we obtain the following.

\begin{proposition}
\label{lemma:amp_min_r}
The competitive ratio of AMP for $k\geq 4$ is upper bounded by the expression
\begin{equation}
\label{eq:inf_AMP_improved}
\min_{r>1} \frac{r^k}{r^{k-1}-r}.
\end{equation}
\end{proposition}
\begin{proof}
    Consider an arbitrary phase $p$ and a fixed parameter $r>1$.  The expression $r^k/(r^{k-1}-r)$ is at least $r$. As observed earlier, at the end of the first $k$ phases, the competitive ratio is at most $r$, hence the proof holds for $p\leq k$.  For the remaining case $p \geq k+1$, 
    we have that
    \begin{align*}
          \frac{|\textrm{OPT}|}{|\textrm{AMP}_p|}
        \leq & \frac{ r^{\ell(p)+1}} { r^{\ell(p)} - r^{\ell(p - k +1) + 1} }
        \\
        \leq & \frac{ r^{\ell(p)+1}} { r^{\ell(p)} - r^{\ell(p) - k +2 } }
        \tag{From~\eqref{eq:monotonicity.phase}}
        \\
        = & \frac{r}{1 - r^{-k+2}}
        \\
        = & \frac{r^k}{r^{k-1} - r}.
    \end{align*}
\end{proof}

Next, we show that Proposition~\ref{lemma:amp_min_r} can yield an improved analysis of AMP over the original bound~\eqref{eq:inf_AMP}.

\begin{proposition}
\label{prop:improvement}
For all even $k\geq 4$, we have that the competitive ratio as expressed by \eqref{eq:inf_AMP} is at least the expression \eqref{eq:inf_AMP_improved}.
\end{proposition}
\begin{proof}
For $k=4$, we obtain numerically that \eqref{eq:inf_AMP} is 
2.64526
whereas~\eqref{eq:inf_AMP_improved} is 2.59808.

For even $k\geq 6$, first we show that the minimizer (for $r>1$) of
\[
\frac{r^k(r-1)}{r^{k-1}(r-1)-r}
\]
is between $r_0$ and $2$.  The derivative of the above expression is
\[
    \frac{(r-1)^2 r^{2 k}+(k-1-kr) r^{k+2}}{\left(r^2-(r-1) r^k\right)^2},
\]
which we claim to be positive for any $r\geq 2$. This follows by the inequality $r^{k-3} \geq k$ which holds for any $r \geq 2$ and $k \geq 6$.

Similarly, for even $k\geq 6$, we show that the minimizer (for $r>1$) of
\[
\frac{r^k}{r^{k-1}-r}
\]
 is between $1$ and $2$. The derivative in $r$ of the above expression is
\[
\frac{r^k \left(r^k-(k-1) r^2\right)}{\left(r^2-r^k\right)^2},
\]
which we claim to be positive for any $r\geq 2$. Again this follows by the inequality $r^{k-2} \geq k-1$, which holds for any $r \geq 2$ and $k \geq 6$.
The proof follows since for $1<r\leq 2$, we have
\[
 \frac{r^k(r-1)}{r^{k-1}(r-1)-r} \geq  \frac{r^k(r-1)}{r^{k-1}(r-1)-r(r-1)} =\frac{r^k}{r^{k-1}-r}.
\]
\end{proof}

The following theorem concludes the asymptotic analysis of the performance of AMP.
\begin{theorem}
For all even $k$, AMP has  competitive ratio $1 + O(\frac{\log k}{k})$.
\label{thm:asymptotic.amp}
\end{theorem}
\begin{proof}
We first sketch a simple argument based on the Puiseux series expansion~\cite{Siegel:complex.function}: this is a type of power series that allows fractional powers, as opposed to only integer ones (e.g., Taylor series). 
Let $r$ denote the optimal choice of the parameter, namely the one that minimizes~\eqref{eq:inf_AMP_improved}.
By analyzing the derivative, it follows that $r=(k-1)^{1/(k-2)}$, hence the  competitive ratio is at most
$
\frac{(k-1)^\frac{k-1}{k-2}}{k-2},
$
whose Puiseux series expansion at $k=\infty$ is
$
1+\frac{\log k+1}{k}+O(\frac{1}{k^2}).
$

For completeness, we give a second proof that relies only on standard calculus. Let $f$ be such that $r=1+f/k$, then since
$r=(k-1)^{1/(k-2)}$, we have that the competitive ratio is at most
\begin{equation}
r \cdot \frac{k-1}{k-2} = \left(1+\frac{f}{k} \right) \frac{k-1}{k-2} = 1+O\left(\frac{f}{k}+\frac{1}{k} \right).
\label{eq:competitive.ratio.amp}
\end{equation}
Suffices then to show that $f=O(\log k)$.
Consider the function
\[
    g(x) = x^{k-2}-k+1,
\]
and note that $r$ must be a root of $g$. We can rewrite $g(r)$ as
\[
    g(r) = \left(1+\frac{f}{k} \right)^{k-2}-k+1 = e^{(k-2)\ln(1+f/k)}-k+1.
\]
Hence, we have
\begin{equation*}
\label{eq:cr2}
g(r) \ge e^{(k-2)2f/(2k+f)} - k + 1.
\end{equation*}
Here we used the following logarithmic inequality~\cite{love:inequalities} for $n\geq 0$:
\[
\ln(1+1/n) \ge \frac{2}{2n+1}.
\]
Suppose now that $f = 4\ln k$. In this case, we claim that $(k-2)\frac{2f}{2k+f} \ge \frac{f}{2}$, or equivalently $2k \ge 8+f$,
which holds for sufficiently large $k$. Thus we have
\[
g\left(1+ \frac{4\ln k}{k} \right) \ge e^{f/2} -k + 1 = k^2 - k + 1 > 0.
\]
On the other hand,
\(
g(0) = -k + 1 < 0.
\)
As a result, since $g$ is continuous, there exists $f \in [0, 4\ln k]$ such that $g(r) = 0$. Therefore, $f = O(\log k)$,
which concludes the proof.
\end{proof}

\subsection{The algorithm \textsc{Greedy}}
\label{sec:greedy}

We consider the algorithm \textsc{Greedy}, which repeatedly applies an
arbitrary augmenting path whenever possible. More precisely, let $E$ denote 
the set of edges that have been revealed to the algorithm, and let $e$ denote
a newly arriving edge. Then as long as there is a non-blocked augmenting path 
in the graph $(V,E\cup\{e\})$, where $V$ is the vertex set, 
\textsc{Greedy} will apply such a path until, 
eventually, no such path any longer exists. Note, in particular, that \textsc{Greedy}
does nothing if no such path exists upon the arrival of $e$.

This algorithm achieves an upper bound of $3/2$ for $k=2$, as shown in \cite{BoyarFavrholdt:17:Relaxing-the-Irrevocability}. 
We show that the same guarantee holds for all even $k$. 
In what concerns the
lower bound, the idea is to force the algorithm to augment an arbitrarily long
path in order to create a configuration with an arbitrarily large number of
blocked augmenting paths of lengths 5, which have local ratio $3/2$.  

\begin{proposition} \label{prop:greedy}
The  competitive ratio of \textsc{Greedy} is $3/2$ if $k$ is even, 
and $2$, if $k$ is odd.
\end{proposition}

\begin{proof}
First note that
\textsc{Greedy} has the property that every edge in the optimal matching has
at least one endpoint matched by \textsc{Greedy}. As a result, the
competitive ratio is at most 2 for any $k$, and in particular for any odd $k$.

For even positive $k$ we give a stronger upper bound of $3/2$. Consider the
symmetric difference of the matching produced by \textsc{Greedy} and an
optimal matching which consists of alternating cycles and paths.  In each of
these components the local ratio is $1$ for the alternating cycles and
alternating paths of even length. We claim that alternating paths of odd
length have length $\ell$ at least 5, and therefore the local ratio is at most
$3/2$. To see this, observe that the case $\ell=1$ corresponds to an edge with
both endpoints unmatched, and \textsc{Greedy} would have included it in its
matching. In the case $\ell=3$, the center edge has an odd budget, meaning
that \textsc{Greedy} would have applied this augmenting path.

The proof of the lower bound for even positive $k$ consists of an instance on
which \textsc{Greedy} achieves the competitive ratio $3/2 - \epsilon$ for any small
constant $\epsilon>0$; see Figure~\ref{fig:lb_greedy}. Let $n$ be a
sufficiently large integer.  First the adversary releases $2n+1$ vertex
disjoint edges, which the algorithm includes in its matching. Then these edges
are connected with $2n+2$ new edges to form an augmenting path of length
$4n+3$. From now on, each time the algorithm applies this augmenting path in
its matching, the path is extended with one additional edge on each end, until
there is at least one edge of type $k$ on the path.  At this point the edge
types on the path form a sequence which starts with types from 1
to $k-1$, then alternates between types $k$ and $k-1$ and finally ends with
types from $k-1$ to $1$. To complete the instance, the adversary attaches
a new edge on each endpoint of every second type $k$ edge.  As a result
there are $n$ augmenting paths of length $5$, which are all blocked by a type
$k$ edge, together with 2 alternating paths of length $k$. The size of the
matching produced by \textsc{Greedy} is $2n+k$, whereas the optimal matching
has size $3n+k$, showing a lower bound of $3/2$.

For odd $k$ the construction can be further strengthened. In the
final step, the adversary attaches an edge on each endpoint of every type $k$ edge, creating
an arbitrary large number of blocked augmenting paths of length $3$.

Note that the lower bounds hold even for the asymptotic competitive ratio, by repeating these
constructions as in the proof of Lemma \ref{lem:lb_asym_LGREEDY}.
\end{proof}

\begin{figure}[ht]
 \centering
 \begin{tikzpicture}[scale=.77,very thick]
    \node[sommet] (v0) at (0,0) {};
    \node[sommet] (v1) at (1,0) {};
    \node[sommet] (v2) at (2,0) {};
    \node[sommet] (v3) at (3,0) {};
    \node[sommet] (v4) at (4,0) {};
    \node[sommet] (v5) at (5,0) {};
    \node[sommet] (v6) at (6,0) {};
    \node[sommet] (v7) at (7,0) {};
    \node[sommet] (v8) at (8,0) {};
    \node[sommet] (v9) at (9,0) {};
    \node[sommet] (v10) at (10,0) {};
    \node[sommet] (v11) at (11,0) {};
    \node[sommet] (v12) at (12,0) {};
    \node[sommet] (v13) at (13,0) {};
    \node[sommet] (v14) at (14,0) {};
    \node[sommet] (v15) at (15,0) {};
    \node[sommet] (v16) at (3,-1) {};
    \node[sommet] (v17) at (4,-1) {};
    \node[sommet] (v18) at (7,-1) {};
    \node[sommet] (v19) at (8,-1) {};
    \node[sommet] (v20) at (11,-1) {};
    \node[sommet] (v21) at (12,-1) {};
    \draw[alg] (v0) -- node[auto] {1} (v1);
    \draw[alg] (v2) -- node[auto] {3} (v3);
    \draw[alg] (v4) -- node[auto] {3} (v5);
    \draw[alg] (v6) -- node[auto] {3} (v7);
    \draw[alg] (v8) -- node[auto] {3} (v9);
    \draw[alg] (v10) -- node[auto] {3} (v11);
    \draw[alg] (v12) -- node[auto] {3} (v13);
    \draw[alg] (v14) -- node[auto] {1} (v15);
    \draw[opt] (v1) -- node[auto] {2} (v2);
    \draw[opt] (v3) -- node[auto] {0} (v16);
    \draw[opt] (v4) -- node[auto] {0} (v17);
    \draw[opt] (v5) -- node[auto] {4} (v6);
    \draw[opt] (v7) -- node[auto] {0} (v18);
    \draw[opt] (v8) -- node[auto] {0} (v19);
    \draw[opt] (v9) -- node[auto] {4} (v10);
    \draw[opt] (v11) -- node[auto] {0} (v20);
    \draw[opt] (v12) -- node[auto] {0} (v21);
    \draw[opt] (v13) -- node[auto] {2} (v14);
    \draw[none] (v3) -- node[auto] {4} (v4);
    \draw[none] (v7) -- node[auto] {4} (v8);
    \draw[none] (v11) -- node[auto] {4} (v12);
\end{tikzpicture}
 \caption{Lower bound construction on the competitive ratio of \textsc{Greedy} for $k = 4$ and $n = 2$. Edges are labeled by their type.  Solid/blue edges depict the algorithm's matching, dashed/red edges depict the optimal matching, dotted/black edges belong to none of the matchings.}
 \label{fig:lb_greedy}
 \end{figure}

\subsection{The algorithm $L$-\textsc{Greedy}}
\label{sec:L-greedy}

How can the greedy algorithm be improved?  As illustrated in the proof of
Proposition~\ref{prop:greedy}, the greedy algorithm has inferior performance because it augments
along arbitrarily long augmenting paths, therefore sometimes sacrificing edge
budget for only a marginal increase in the matching size.  A natural idea
towards an improvement would be to apply only short augmenting paths, as they
are more budget efficient.  For technical reasons, we restrict the choice of
augmenting paths even further.

We define the algorithm $L$-\textsc{Greedy} for some given
parameter $L$, which applies any non-blocked augmenting path of length at most
$2L+1$ that is in the symmetric difference between the current matching and
some particular optimal matching OPT.  The latter is updated after each edge
arrival by applying an augmenting path to OPT. 

To make the above more precise, let $E$ denote 
the set of edges that have been revealed to the algorithm, and let $e$ denote
a newly arriving edge. First, we explain how the optimal matching is updated: 
If $\textrm{OPT}(E)$ denotes the optimal matching after all edges in $E$ have 
been revealed, then  $\textrm{OPT}(E\cup \{e\})$ is obtained by applying an 
augmenting path to $\textrm{OPT}(E)$. Then, $L$-\textsc{Greedy} serves request
$e$ by consecutively applying non-blocked augmenting path of length at most
$2L+1$  that is in the symmetric difference between its current matching (prior to the arrival of $e$)
and $\textrm{OPT}(E)$.

Note that $L$-\textsc{Greedy} may not change its solution even if there is a
short augmenting path in the current graph if it contains edges
which are not in this particular optimal matching OPT.
We will later optimize the parameter $L$ as a function of $k$.

\subsubsection{Analysis of $L$-\textsc{Greedy}}

We begin by observing that for $L=0$ the algorithm collects greedily vertex
disjoint edges without any recourse, which is precisely the behavior of
\textsc{Greedy} for $k=1$ and has  competitive ratio $2$. For $L=1$ the
algorithm $L$-\textsc{Greedy} applies only augmenting paths of length at most
3. In this case,
the same argument as
in the proof of Proposition~\ref{prop:greedy} shows that the
competitive ratio of $L$-\textsc{Greedy} is $3/2$.

In what follows we analyze the general case $L\geq 2$. To this end, we assign
weights to vertices in a way that the total vertex weight equals the size
of the current matching.  Therefore, whenever the size of the matching is
increased by $1$, a total weight of $1$ is distributed on the vertices along
the augmenting path as follows. First, vertices in this path that were already matched receive
a weight $\alpha$, where $\alpha \geq 0$ is some constant that we specify
later.  Second, the two vertices on the endpoints of the augmenting path
receive the remaining weight, that is $1/2 - \ell \alpha$, where $2\ell+1$ is
the length of the path. It follows, from this weight assignment, that every
unmatched vertex has weight $0$, that every matched vertex has  weight at
least $1/2-L\alpha$, and that  every endpoint of a type $k$ edge has weight at
least $1/2-L\alpha + (k-1)\alpha$.

Suppose that $L$-\textsc{Greedy} reaches a configuration in which it cannot
apply any augmenting path, as specified in its statement. We consider the
symmetric difference between the matching produced by the algorithm and the
optimal matching maintained internally by the algorithm. This symmetric
difference consists of alternating paths and/or alternating cycles, and we
will upper bound for each such component its local ratio. In particular, a
component in the symmetric difference falls in one of the following cases:
Either it is an augmenting path of length $2\ell+1 \leq 2L+1$, or an
augmenting path of length  $2\ell+1 >2L+1$, or an alternating cycle or
alternating path of even length.

\bigskip
\noindent
{\em Case 1: Augmenting path of length $2\ell+1\leq 2L+1$.} \
Note that such a path contains at least one edge of type $k$, otherwise the algorithm 
would augment it. It follows that $\ell\geq
2$, since an augmenting path of length $1$ is a single type 0 edge, and an
augmenting path of length $3$ has edge types respectively $0,t,0$ for some odd
$t$ (and $k$ is assumed to be even).  The path contains $2\ell$ matched
vertices, and at least $2$ of them are incident with a type $k$ edge.  Hence the
total vertex weight is at least
$
    2\ell\left(\frac12 - L \alpha\right) + 2(k-1)\alpha,
$
and the local ratio of this component is at most
\begin{equation}\label{expr:short}
 \frac{\ell+1}{\ell - 2\ell L \alpha + 2(k-1)\alpha}.
\end{equation}

\bigskip
\noindent
{\em Case 2: Augmenting path of length $2\ell+1> 2L+1$.} \
Such a path contains $2\ell$ matched vertices and therefore the local ratio is at most
\begin{equation}\label{expr:long}
    \frac{\ell+1}{\ell - 2\ell L \alpha}.
\end{equation}

\bigskip
\noindent
{\em Case 3: Alternating cycle or path of even length.} \
Such a component contains $2\ell$ matched vertices and therefore the local ratio is at most
\begin{equation*}\label{expr:cycle}
    \frac{\ell}{\ell - 2\ell L \alpha}, 
\end{equation*}
\medskip
which is dominated by \eqref{expr:long}.
We obtain the following performance guarantee.

\begin{theorem}
The  competitive ratio of $L$-\textsc{Greedy} with $L = \lfloor \sqrt{k-1} \rfloor$ is at most
\[
\frac{k(L +2) - 2}{(L +1)(k-1)}
= 1+O\left( \frac{1}{\sqrt{k}} \right),
\]
for even $k\geq 6$ and at most $3/2$ for $k=4$.
\label{thm:upper.l_greedy}
\end{theorem}
\begin{proof}
We choose $\alpha$ so as to minimize the maximum of the local ratios, as defined by~\eqref{expr:short} and \eqref{expr:long}. Then for this choice of $\alpha$ we optimize $L$ as stated in the theorem.
Note that for $k=4$, this leads to the choice $L=1$, which we analyzed in the beginning of the section.

For $k\geq 6$, we have to minimize over $\alpha$ and $L$ the maximum over $\ell$ of the two
ratios given by~\eqref{expr:short} and~\eqref{expr:long}.
First we upper bound~\eqref{expr:long} as
\begin{equation}\label{expr:long2}
    \frac{\ell+1}{\ell - 2\ell L \alpha} \leq
    \frac{L+2}{L+1 - 2L^2 \alpha - 2L \alpha},
\end{equation}
where the inequality follows from $\ell\geq L+1$ and the fact that the
left hand side is decreasing in $\ell$ as can be seen by dividing
both the numerator and denominator by $\ell$.

In order to upper bound \eqref{expr:short}, we find its derivative in $\ell$ which is equal to
\[
    \frac{2 \alpha (L+k-1) - 1}{(\ell+2\alpha(k-L\ell-1))^2}.
\]
This means that (\ref{expr:short}) is increasing or decreasing in $\ell$ depending on the sign of $2 \alpha (L+k-1)-1$. Hence we distinguish two cases.

\paragraph{Case 1: $\alpha < 1/(2(L+k-1))$}

In this case (\ref{expr:short}) is decreasing in $\ell$, and by $\ell\geq 2$ is at most
\begin{equation}\label{expr:short2}
 \frac{3}{2\alpha (k-2L-1) + 2}.
\end{equation}

\paragraph{Subcase 1a: $k-2L - 1 < 0$}

In this case (\ref{expr:short2}) and (\ref{expr:long2}) are increasing in
$\alpha$ and hence minimized at $\alpha=0$. For this choice of $\alpha$ (\ref{expr:short2})
is $3/2$, while (\ref{expr:long2}) is $(L+2)/(L+1)$ which by $L\geq 2$ is
less than $3/2$.  In conclusion, the competitive ratio in Subcase 1a is at most $3/2$.

\paragraph{Subcase 1b: $k-2L-1 \geq 0$}

In this case (\ref{expr:short2}) is non-increasing in $\alpha$ while (\ref{expr:long2}) is increasing in $\alpha$.  Hence the maximum of (\ref{expr:short2}) and (\ref{expr:long2}) is minimized at the equality of the expressions, which happens for
\[
\alpha = \frac {L-1}{2L^2 + 4k + 2kL -4L -4}.
\]
It can readily be verified that this choice of $\alpha$ indeed satisfies the assumption of Case~1 for all even $k\geq 2$ and $L\geq 2$.
The corresponding ratio is
\begin{equation}               \tag{R1b} \label{expr:1b}
\frac{L^2 + 2k + kL -2L -2}{(k-1)(L+1)}.
\end{equation}

\paragraph{Case 2: $\alpha \ge 1/(2(L+k-1))$}
In this case (\ref{expr:short}) is increasing in $\ell$, and at $\ell=L$ becomes
\begin{equation}\label{expr:short3}
\frac{L+1}{L+2 \alpha(k-L^2 - 1)}.
\end{equation}

\paragraph{Subcase 2a: $k-L^2 -1 < 0$}  In this case, \eqref{expr:short3} and \eqref{expr:long2} are increasing in $\alpha$.  At $\alpha=1/(2(L+k-1))$, the former becomes
\begin{align}
\frac{L+1}{L+\frac{k-L^2-1}{L+k-1}}
= & \frac{(L+1)(L+k-1)}{L(L+k-1)+k-L^2-1}  \notag
\\
= & \frac{(L+1)(L+k-1)}{(L+1)(k-1)}        \notag
\\
= & \frac{L+k-1}{k-1}.                     \label{expr:2aa}
\end{align}
Similarly, \eqref{expr:long2} becomes
\begin{align}
\frac{L+2}{L+1-\frac{L(L+1)}{L+k-1}}
= & \frac{(L+2)(L+k-1)}{(L+1)(L+k-1)-L(L+1)}    \notag
\\
= & \frac{(L+2)(L+k-1)}{(L+1)(k-1)},      \tag{R2a} \label{expr:2a}
\end{align}
which dominates \eqref{expr:2aa} and therefore upper bounds the competitive ratio in Subcase 2a.

\paragraph{Subcase 2b: $k -L^2 - 1 \geq 0$} In this case, \eqref{expr:short3} is non-increasing in $\alpha$ while \eqref{expr:long2} is increasing in $\alpha$.  We have equality of the expressions for
\[
\alpha = \frac{1}{2k(L+2) - 4},
\]
which implies that the competitive ratio in Subcase 2b is at most
\begin{equation}                                    \tag{R2b}    \label{expr:2b}
\frac{k(L+2) - 2}{(L+1)(k-1)}.
\end{equation}
This concludes the case analysis. Thus it remains to optimize $L$.

For $k=4$, the assumptions of the Subcases 1b and 2b are not valid.  Hence, we have to choose $L$ so to minimize the minimum of $3/2$ and \eqref{expr:2a}, which is larger than $3/2$. This means that any choice $L\geq 2$ leads to a competitive ratio at most $3/2$. Note that the same guarantee is obtained for $L=1$.

Finally we consider $k\geq 6$ and choose $L$ so to minimize the minimum of $3/2$ and the ratios~\eqref{expr:1b}, \eqref{expr:2a}, \eqref{expr:2b}.
We claim that \eqref{expr:1b} is dominated by $3/2$. Indeed its derivative in $k$ is
\[
-\frac{L(L-1)}{(k-1)^2 (L+1)} < 0.
\]
Hence \eqref{expr:1b} is maximized for $k,L$ such that $k-2L-1 = 0$, and this maximum value is precisely $3/2$.

By comparing the numerators of the ratios~\eqref{expr:1b}, \eqref{expr:2a} and \eqref{expr:2b} the minimum ratio is \eqref{expr:2b}. Its derivative in $L$ is
\[
-\frac{k-2}{(L+1)^2(k-1)} \leq 0.
\]
Hence the minimum of \eqref{expr:2b} is attained at the upper bound for $L$ given by the Subcase 2b assumption, namely
\[
       L = \lfloor \sqrt{k-1} \rfloor.
\]
It follows that the competitive ratio is at most
\[
\frac{k(\lfloor \sqrt{k-1} \rfloor+2) - 2}{(\lfloor \sqrt{k-1} \rfloor+1)(k-1)} \leq 1 + \frac{2k + \sqrt{k-1} -1}{(\sqrt{k-1} +1)(k-1)}
= 1 + O\left(\frac1{\sqrt{k}}\right).
\]
\end{proof}

We complete the analysis of $L$-\textsc{Greedy} by showing that the upper bound is essentially tight.

\begin{lemma}                   \label{lem:lb_LGREEDY}
For even $k\geq 4$, the strict competitive ratio of $L$-\textsc{Greedy} with $L=\lfloor \sqrt{k-1}\rfloor$ is at least
\[
        1 + \Omega\left(\frac{1}{\sqrt{k}}\right).
\]
\end{lemma}

\begin{proof}
For the proof we show that for any positive parameter $L\geq 3$ the (strict) competitive ratio of $L$-\textsc{Greedy} is at least
\begin{equation}
\label{eq:lb_LGREEDY}
\frac{3 \lfloor \frac{L-1}{2} \rfloor + k-2}{L + k-3}.
\end{equation}
The statement follows then by lower bounding this expression for $L=\lfloor \sqrt{k-1} \rfloor$.

We will show how to create an instance in which $L$-\textsc{Greedy} has the
competitive ratio expressed by \eqref{eq:lb_LGREEDY}. The adversarial
construction is depicted in Figure~\ref{fig:lb_LGREEDY}, and consists of the
following steps. First the adversary releases $L-2$ vertex disjoint edges,
denoted by $e_1, \ldots, e_{L-2}$, which the algorithm includes in its
matching. Then these edges are connected with $L-1$ new edges (which are denoted by
$f_1, f_2, \ldots ,f_{L-1}$ in Figure~\ref{fig:lb_LGREEDY}) so as to form an
augmenting path $P$ of length $2L-3$, which is short enough to ensure that the
algorithm applies it. Note that the precise order in which the edges arrive is not
important. At this stage the edges of $P$ have alternating types
$1$ and $2$, and by the choice of $L\geq 3$ there is at least one edge of type
$2$ in $P$.

In the next phase, the adversary releases edges from the sets $A$ and $B$, in
an interleaved manner; specifically, $P$ is extended with two additional edges
on each end, so as to form augmenting paths of length $2L-1$ and $2L+1$, until
there are edges of type $k$ on the path. At this point, the edge types on the
medium part of length $2L-3$ form an alternating sequence of $k-1$ and $k$,
and note that there are $k/2-1$ edges of type $2$ and $k/2-1$ edges of type
$1$ in each set $A$ and $B$. To complete the instance, the adversary attaches
a new edge on each of the type $k-1$ edges in $P$, alternating between the
left endpoint and on the right endpoint of $P$. As a result, all augmenting
paths are blocked with a type $k$ edge. The size of the matching produced by
$L$-\textsc{Greedy} is $L+ k - 3$, whereas the optimal matching has size at
least $3 \lfloor \frac{L-1}{2} \rfloor + k-2$, concluding the proof.
\end{proof}

\begin{figure}[ht]
 \centering
 \begin{tikzpicture}[very thick]
    \draw[fill=gray!20,rounded corners,draw=white] (-2,3) rectangle (0,-1.5);
    \draw[fill=gray!20,draw=black] (-1,2.7) node {A};
    \draw[fill=gray!20,rounded corners,draw=white] (9,3) rectangle (11,-1.5);
    \draw[fill=gray!20,draw=black] (10,2.7) node {B};

    \node[sommet] (v0) at (0,0) {};
    \node[sommet] (v1) at (1,0) {};
    \node[sommet] (v2) at (2,0) {};
    \node[sommet] (v3) at (3,0) {};
    \node[sommet] (v4) at (4,0) {};
    \node[sommet] (v5) at (5,0) {};
    \node[sommet] (v6) at (6,0) {};
    \node[sommet] (v7) at (7,0) {};
    \node[sommet] (v8) at (8,0) {};
    \node[sommet] (v9) at (9,0) {};


    \node[sommet] (v10) at (-2,2) {};
    \node[sommet] (v11) at (-1,2) {};
    \node[sommet] (v12) at (-2,1) {};
    \node[sommet] (v13) at (-1,1) {};
    \node[sommet] (v14) at (-2,-1) {};
    \node[sommet] (v15) at (-1,-1) {};

    \draw (-1.5,0) node[above] {$\vdots$};

    \node[sommet] (v18) at (10,2) {};
    \node[sommet] (v19) at (11,2) {};
    \node[sommet] (v20) at (10,1) {};
    \node[sommet] (v21) at (11,1) {};
    \node[sommet] (v22) at (10,-1) {};
    \node[sommet] (v23) at (11,-1) {};

 \draw (10.5,0) node[above] {$\vdots$};

    \node[sommet] (v26) at (0.1,-1) {};
    \node[sommet] (v27) at (3,-1) {};
    \node[sommet] (v29) at (8,-1) {};
    \node[sommet] (v30) at (4,-1) {};
    \node[sommet] (v31) at (7,-1) {};

    \draw[alg] (v0) -- node[auto] {k-1} (v1);
    \draw[alg] (v2) -- node[auto] {k-1} (v3);
    \draw[alg] (v6) -- node[auto] {k-1} (v7);
    \draw[alg] (v8) -- node[auto] {k-1} (v9);
    \draw[alg] (v4) -- node[auto] {k-1} (v5);

    \foreach \a/\b in {v10/v11, v12/v13, v14/v15, v18/v19, v20/v21, v22/v23} {
        \draw[alg] (\a) -- node[auto] {1} (\b);
        \draw[opt,transform canvas={yshift=-2pt}] (\a)  -- (\b);
    };



    \draw[opt] (v1) -- node[auto] {k} (v2);
    \draw[opt] (v1) -- node[below] {$e_1$} (v2);
    \draw[opt] (v0) -- node[right,yshift=-1ex] {0} (v26);
    \draw[opt] (v3) -- node[auto,yshift=-0.5ex] {0} (v27);
    \draw[opt] (v8) -- node[left,yshift=-0.5ex] {0} (v29);
    \draw[opt] (v4) -- node[left,yshift=-0.5ex] {0} (v30);
    \draw[opt] (v7) -- node[right,yshift=-0.5ex] {0} (v31);
    \draw[opt] (v5) -- node[auto] {k} (v6);
    \draw[opt] (v5) -- node[below] {$e_3$} (v6);

    \draw[alg] (v0) -- node[below] {$f_1$} (v1);
    \draw[alg] (v2) -- node[below] {$f_2$} (v3);
    \draw[alg] (v4) -- node[below] {$f_3$} (v5);
    \draw[alg] (v6) -- node[below] {$f_4$} (v7);
    \draw[alg] (v8) -- node[below] {$f_{L-1}$} (v9);

    \draw[none] (v7) -- node[auto] {k} (v8);
    \draw[none] (v7) -- node[below] {$e_{L-2}$} (v8);

    \draw[none] (v3) -- node[auto] {k} (v4);
    \draw[none] (v3) -- node[below] {$e_2$} (v4);

    \draw[none] (v0) -- node[midway,yshift=3ex] {2} (v11);
    \draw[none] (v0) -- node[midway,yshift=-2ex] {2} (v13);
    \draw[none] (v0) -- node[midway,yshift=-2ex] {2} (v15);
    \draw[none] (v9) -- node[midway,yshift=3ex] {2} (v18);
    \draw[none] (v9) -- node[midway,yshift=-2ex] {2} (v20);
    \draw[none] (v9) -- node[midway,yshift=-2ex] {2} (v22);

    \draw[<->](-2.6,2) -- (-2.6,-1) node[midway,fill=white] {k/2-1};
    \draw[<->](11.6,2) -- (11.6,-1) node[midway,fill=white] {k/2-1};

    \draw[<->](0,1) -- (9,1) node[midway,fill=white] {2L-3};

\end{tikzpicture}
 \caption{Lower bound construction for even $k$ and $L = 6$. Numbers on edges describe the edge types at the end of the
 algorithm's execution. Solid/blue edges depict the algorithm's matching, dashed/red edges depict the optimal matching and dotted/black edges belong to none of the matchings.}
 \label{fig:lb_LGREEDY}
 \end{figure}

The previous lemma can be extended to the asymptotic competitive ratio, using a standard technique based on multiple copies of the adversarial instance.
\begin{lemma}                \label{lem:lb_asym_LGREEDY}
For even $k$ and $L=\lfloor \sqrt{k-1} \rfloor \geq 3$, the asymptotic competitive ratio of $L$-\textsc{Greedy} is at least the expression~\eqref{eq:lb_LGREEDY} and hence is $1+\Omega(1/\sqrt k)$.
\end{lemma}
\begin{proof}
    Assume that the asymptotic competitive ratio of $L$-\textsc{Greedy} is $R$ for $R$ being strictly smaller than expression~\eqref{eq:lb_LGREEDY}.  Then by definition there exists a constant $d$ such that
    \[
        |L\textsc{-Greedy}(\sigma)| \geq |\textrm{OPT}(\sigma)| / R - d
    \]
    for all request sequences $\sigma$.  Let $\sigma'$ be the adversarial construction defined in the proof of Lemma~\ref{lem:lb_LGREEDY}. For an arbitrary positive integer $p$ let $\sigma$ be the result of repeating each edge of $\sigma'$ $p$ times, such that the resulting graph consists of $p$ disjoint copies of the graph produced by $\sigma'$.  By the above assumption we have
    \begin{align*}
      |L\textsc{-Greedy}(\sigma)| &\geq |\textrm{OPT}(\sigma)| / R - d   &\equiv
      \\
      p\cdot |L\textsc{-Greedy}(\sigma')| &\geq p\cdot |\textrm{OPT}(\sigma')| / R - d
      &\equiv
      \\
      |L\textsc{-Greedy}(\sigma')| &\geq |\textrm{OPT}(\sigma')| / R - d/p.
    \end{align*}
    Since $d/p$ can be arbitrarily close to $0$, this would mean that $L$-\textsc{Greedy} is $R$-competitive, a contradiction.
\end{proof}

\subsection{Lower bound on the competitive ratio of deterministic algorithms}
\label{sec:(k-1)/k}

Boyar \etal \cite{BoyarFavrholdt:17:Relaxing-the-Irrevocability} show that the deterministic competitive ratio of the problem is $2$ for $k=1$ and $3/2$ for $k=2$. We complete this picture by showing a lower bound of $1+\frac{1}{k-1}$ for all $k\geq 3$. Note that the lower bound is tight for $k=3$, as the algorithm \textsc{Greedy}, which works by assuming that $k$ is only $2$, has competitive ratio $3/2$.

\begin{theorem}
    The deterministic competitive ratio of the online matching problem with $k$ edge-recourse is at least $1+\frac{1}{k-1}$ for all $k \ge 3$.
\label{thm:lower.bound.deterministic}
\end{theorem}

\begin{proof}
We consider three cases, namely the cases $k=3$, $k$ is even and at least 4, and finally $k$ is odd and at least 5.
For each case we present an appropriate adversarial argument. 

\paragraph{Case $k=3$.}
Suppose, by way of contradiction, some algorithm claims a
competitive ratio strictly smaller than $(3n+2)/(2n+2)$ for some arbitrary
$n\geq 1$. The adversary releases a single edge, creating an augmenting path
of length $1$.  Then the algorithm applies the augmenting path, which the
adversary extends by appending one edge on each side, creating an augmenting
path of type 0,1,0, as shown in Figure~\ref{fig:lower.bound}(a).  Since the
current competitive ratio is $2$, the algorithm needs to apply this path, which the
adversary again extends by appending an edge on each side, creating an
augmenting path of type 0,1,2,1,0, as shown in
Figure~\ref{fig:lower.bound}(b). Since the current competitive ratio is $3/2$, the
algorithm applies this path. In response the adversary appends an edge at each
endpoint of the type 3 edge, and at each endpoint of one of the type 1 edges,
as shown in Figure~\ref{fig:lower.bound}(c).  The resulting graph has a
blocked augmenting path of type 0,3,0, and an augmenting path of type 0,1,0,
as shown in Figure~\ref{fig:lower.bound}(c). The algorithm needs to apply the
latter one as the competitive ratio is currently $5/3 > 3/2$.

\begin{figure}[ht]
\begin{tikzpicture}[very thick]
\draw[rounded corners,fill=gray!20,draw=white] (10.7,-0.3) rectangle (13.3,1.8);
\draw (0,2) node {(a)};
\draw (4,2) node {(b)};
\draw (8.3,2) node {(c)};
\foreach \x/\y/\n in {0/1/a, 1/1/b, 2/1/c, 0/0/d, 4/1/e, 4/0/f, 5/1/g, 5/0/h, 6/1/i, 7/1/j, 9/1/k, 9/0/l, 9/2/n, 10/1/p, 10/2/o, 10/0/m, 11/1/q, 11/0/r, 12/1/s, 13/1/t}
        \node[sommet] (\n) at (\x,\y) {};
\draw[alg]  (a) -- node[auto] {1} (b);
\draw[opt]  (a) -- node[auto] {0} (d);
\draw[opt]  (b) -- node[auto] {0} (c);

\draw[opt]  (f) -- node[auto] {0} (h);
\draw[alg]  (f) -- node[auto] {1} (e);
\draw[opt]  (e) -- node[auto] {2} (g);
\draw[alg]  (g) -- node[auto] {1} (i);
\draw[opt]  (i) -- node[auto] {0} (j);

\draw[alg]  (l) -- node[below] {1} (m);
\draw[opt] (11.1,-0.07) -- (11.9,-0.07);
\draw[none] (l) -- node[auto] {2} (k);
\draw[opt]  (k) -- node[auto] {0} (n);
\draw[alg]  (k) -- node[below] {3} (p);
\draw[opt]  (p) -- node[auto] {0} (o);
\draw[none] (p) -- node[auto] {2} (q);
\draw[opt]  (r) -- node[right] {0} (q);
\draw[alg]  (q) -- node[auto] {1} (s);
\draw[opt]  (s) -- node[auto] {0} (t);
\end{tikzpicture}
\caption{Lower bound construction on the competitive ratio for the case $k=3$.}
 \label{fig:lower.bound}
\end{figure}

At this point, the adversary repeats this construction $n-1$ times, by identifying the
shaded part of Figure~\ref{fig:lower.bound}(c) as the graph of
Figure~\ref{fig:lower.bound}(a), and reapplying the above construction.
The final graph consists of $n$ blocked augmenting paths of type 0,3,0 and $n+2$ edges of type $1$ that belong both to the
optimal and the algorithm's matchings. Hence, the competitive ratio is
\(
        ({3n+2})/({2n+2}),
\)
which contradicts the claimed ratio and shows a lower bound on the competitive ratio of $3/2$.

\paragraph{Case $k$ is even and at least 4.} Fix an algorithm that claims
a competitive ratio strictly smaller than $k/(k-1)$.  The adversary releases a single
edge, creating an augmenting path of length $1$.  Whenever the algorithm
applies the augmenting path\footnote{
We can assume, without loss of generality, that this is the only viable choice for the online algorithm. 
This is because the only way an algorithm can transform a given matching $M_1$ to a matching $M_2$ is
via a sequence which can only consist of the following: i) augmenting paths; ii) alternating cycles; and iii)
alternating, or even ``decreasing'' paths (i.e., paths such that if the algorithm applies them, then 
the cardinality of the matching remains the same, or decreases, respectively). 
This is a well-known result from matching theory. In principle an online algorithm, say $A$, could apply paths and cycles in 
cases (ii) and (iii), but such an algorithm can be converted to another algorithm $A'$ which is at least as good as $A$
in terms of matching size, and which maintains edges of smaller types than $A$.
}, the adversary extends it by appending one edge on
each end, eventually creating an alternating path of type $1,2,\dots,k-1, k,
k-1, \dots, 2,1$.  Then the adversary appends an edge to each endpoint of the
type $k-1$ edges. The resulting graph has two augmenting paths of type
$0,k-1,0$, see Figure~\ref{fig:LB}(a).  The algorithm needs to apply them as
the competitive ratio is currently $(k+2)/k$, which is strictly greater than $k/(k-1)$ if
$k\ge 4$. Each augmentation is responded, by the adversary, with an extension
of the path resulting in the  configuration depicted in Figure~\ref{fig:LB}(b)
of ratio $(k+4)/(k+2) \geq k/(k-1)$ where all augmenting paths are blocked by type $k$
edges. Hence the competitive ratio of the algorithm is not strictly smaller
than $k/(k-1)$.

\begin{figure}[h!]
\begin{tikzpicture}[scale=1,very thick,rotate=90]
\foreach \x/\y/\n in {0/2/a, 1/2/b, 2/2/c, 3/2/d, 4/2/e, 5/2/f, 6/2/g, 7/2/h, 8/2/i, 9/2/j, 10/2/k, 11/2/l, 12/2/m, 13/2/n,
                                                         5/1/o, 6/1/p, 7/1/q, 8/1/r}
        \node[sommet] (\n) at (\x,\y) {};
\foreach \u/\v/\col/\type in {a/b/alg/1,
                              b/c/none/2,
                              c/d/alg/3,
                              e/f/none/k-2,
                              f/g/alg/k-1,
                              g/h/none/k,
                              h/i/alg/k-1,
                              i/j/none/k-2,
                              k/l/alg/3,
                              l/m/none/2,
                              m/n/alg/1,
                              f/o/opt/0,
                              g/p/opt/0,
                              h/q/opt/0,
                              i/r/opt/0
                             }
    \draw[\col] (\u) -- node[auto] {\type} (\v);
\draw[Snake] (3,2) -- (4,2);
\draw[Snake] (9,2) -- (10,2);
\draw[opt]   (0.1, 1.95) -- (0.9, 1.95);
\draw[opt]   (2.1, 1.95) -- (2.9, 1.95);
\draw[opt]   (10.1, 1.95) -- (10.9, 1.95);
\draw[opt]   (12.1, 1.95) -- (12.9, 1.95);
\draw (-0.5,2.5) node {(a)};
\end{tikzpicture}
\hspace{1em}
\begin{tikzpicture}[scale=1,very thick,rotate=90]
\foreach \x/\y/\n in {0/2/a, 1/2/b, 2/2/c, 3/2/d, 4/2/e, 5/2/f, 6/2/g, 7/2/h, 8/2/i, 9/2/j, 10/2/k, 11/2/l, 12/2/m, 13/2/n,
                                                         5/1/o, 6/1/p, 7/1/q, 8/1/r,
                                                         5/0/s, 6/0/t, 7/0/u, 8/0/v}
        \node[sommet] (\n) at (\x,\y) {};
\foreach \u/\v/\col/\type in {a/b/alg/1,
                              b/c/none/2,
                              c/d/alg/3,
                              e/f/none/k-2,
                              f/g/opt/k,
                              g/h/none/k,
                              h/i/opt/k,
                              i/j/none/k-2,
                              k/l/alg/3,
                              l/m/none/2,
                              m/n/alg/1,
                              f/o/alg/1,
                              g/p/alg/1,
                              h/q/alg/1,
                              i/r/alg/1,
                              o/s/opt/0,
                              p/t/opt/0,
                              q/u/opt/0,
                              r/v/opt/0
                             }
    \draw[\col] (\u) -- node[auto] {\type} (\v);
\draw[Snake] (3,2) -- (4,2);
\draw[Snake] (9,2) -- (10,2);
\draw[opt]   (0.1, 1.95) -- (0.9, 1.95);
\draw[opt]   (2.1, 1.95) -- (2.9, 1.95);
\draw[opt]   (10.1, 1.95) -- (10.9, 1.95);
\draw[opt]   (12.1, 1.95) -- (12.9, 1.95);
\draw (-0.5,2.5) node {(b)};
\end{tikzpicture}
\hspace{1em}
\begin{tikzpicture}[scale=0.9,very thick,rotate=90]
\foreach \x/\y/\n in {-1/2/aa, 0/2/a, 1/2/b, 2/2/c, 3/2/d, 4/2/e, 5/2/f, 6/2/g, 7/2/h, 8/2/i, 9/2/j, 10/2/k, 11/2/l, 12/2/m, 13/2/n,
                                                                                                                        14/2/nn,
                                                                  5/1/o, 6/1/p, 7/1/q, 8/1/r}
        \node[sommet] (\n) at (\x,\y) {};
\foreach \u/\v/\col/\type in {a/b/alg/1,
                              b/c/opt/2,
                              c/d/alg/3,
                              e/f/alg/k-2,
                              f/g/none/k-1,
                              g/h/alg/k,
                              h/i/none/k-1,
                              i/j/alg/k-2,
                              k/l/alg/3,
                              l/m/opt/2,
                              m/n/alg/1,
                              g/p/opt/0,
                              h/q/opt/0,
                              aa/a/opt/0,
                              n/nn/opt/0,
                              f/o/opt/0,
                              i/r/opt/0
                             }
    \draw[\col] (\u) -- node[auto] {\type} (\v);
\draw[Snake] (3,2) -- (4,2);
\draw[Snake] (9,2) -- (10,2);
\draw (-1,2.7) node {(c)};
\end{tikzpicture}
\hspace{1em}
\begin{tikzpicture}[scale=0.9,rotate=90,very thick]
\foreach \x/\y/\n in {-1/0/a0, -1/1/a1, -1/2/aa, 0/2/a, 1/2/b, 2/2/c, 3/2/d, 4/2/e, 5/2/f,
                      6/2/g, 7/2/h, 8/2/i, 9/2/j, 10/2/k, 11/2/l, 12/2/m, 13/2/n, 14/2/nn,
                      14/1/n1, 14/0/n0,
                      5/0/o0, 8/0/r0,
                      5/-1/o1, 8/-1/r1,
                                                                  5/1/o, 6/1/p, 7/1/q, 8/1/r}
        \node[sommet] (\n) at (\x,\y) {};
\foreach \u/\v/\col/\type in {a/b/alg/3,
                              b/c/opt/4,
                              c/d/alg/5,
                              e/f/alg/k,
                              f/g/none/k-1,
                              g/h/alg/k,
                              h/i/none/k-1,
                              i/j/alg/k,
                              k/l/alg/5,
                              l/m/opt/4,
                              m/n/alg/3,
                              g/p/opt/0,
                              h/q/opt/0,
                              aa/a/opt/2,
                              n/nn/opt/2,
                              f/o/opt/2,
                              i/r/opt/2,
                              aa/a1/alg/1,
                              a1/a0/opt/0,
                              nn/n1/alg/1,
                              n1/n0/opt/0,
                              o/o0/alg/1,
                              r/r0/alg/1,
                              o0/o1/opt/0,
                              r0/r1/opt/0
                             }
    \draw[\col] (\u) -- node[auto] {\type} (\v);
\draw[Snake] (3,2) -- (4,2);
\draw[Snake] (9,2) -- (10,2);
\draw (-1,2.7) node {(d)};
\end{tikzpicture}

\caption{Lower bound constructions for the deterministic competitive ratio. Solid/blue edges depict the algorithm's matching, dashed/red edges depict the optimal matching, dotted/black edges belong to none of the matchings and wiggled lines represent parts of the graph that are contracted for readability.}
\label{fig:LB}
\end{figure}

\paragraph{Case $k$ is odd and at least 5.} Fix an algorithm that claims a
competitive ratio strictly smaller than $k/(k-1)$.  The adversary proceeds as in the
previous case, until the graph consists of a path of type $1,2,\dots,k-1, k,
k-1, \dots, 2,1$. This time, the adversary appends one edge to each endpoint of
the type $k-1$ edges, but also appends one edge at each endpoint of the path.
As a result, there are two augmenting paths of type $0,1,2,\dots,k-2,0$ and a
single blocked augmenting path of type $0,k,0$, see Figure~\ref{fig:LB}(c).

The algorithm needs to apply an augmenting path as the competitive ratio is currently
$(k+3)/k$, which is strictly greater than $k/(k-1)$ for $k\ge 5$. The
adversary responds each augmentation of a path by appending an edge on both
ends of this path. At this moment, the ratio decreased slightly, but still
exceeds the claimed ratio, forcing the algorithm to continue augmenting.
Eventually this leads to a configuration formed by two blocked augmenting
paths of  type $0,1,2,\dots,k, 2, 1, 0$ and a blocked augmenting path of type
$0, k ,0$, see Figure~\ref{fig:LB}(d).  The competitive ratio of this configuration is
$(k+7)/(k+4)$ which exceeds $k/(k-1)$ for $k \ge 5$.
\end{proof}

\subsection{Comparing the algorithms $L$-\textsc{Greedy} and AMP}

We have analyzed two deterministic online algorithms: the algorithm AMP, which has  competitive ratio $1 +
O \left( \log k / k \right)$, and the algorithm $L$-\textsc{Greedy}, which has competitive ratio
$1+\Theta(1/\sqrt k)$. Since the analysis of $L$-\textsc{Greedy} is tight, it follows that AMP is asymptotically
(i.e., for large $k$) superior to $L$-\textsc{Greedy}. However, for small values of $k$, namely $k \leq 20$,
we observe that $L$-\textsc{Greedy} performs better, in comparison to the performance bound we have shown for AMP. These findings are
summarized in Table~\ref{table:upper_bounds} and Figure~\ref{fig:compare_ratios} (Section~\ref{subsec:contribution}).

\begin{table}[htb!]
\small
\begin{center}
\begin{tabular}{ |l|l|l|l|l| }
  \hline
k &LB (arr.) & LB (arr./dep.) & $L$-\textsc{Greedy} & AMP
\\
  \hline
 4 & 1.333333 & 1.428571 & 1.5      & 2.598076 \\
 6 & 1.2      & 1.263158 & 1.466667 & 1.869186 \\
 8 & 1.142857 & 1.179487 & 1.428571 & 1.613602 \\
10 & 1.111111 & 1.134328 & 1.333333 & 1.480583 \\
12 & 1.090909 & 1.106796 & 1.318182 & 1.398080 \\
14 & 1.076923 & 1.088435 & 1.307692 & 1.341500 \\
16 & 1.066666 & 1.075377 & 1.300000 & 1.300080 \\
18 & 1.058823 & 1.065637 & 1.247059 & 1.268330 \\
20 & 1.052631 & 1.058104 & 1.242105 & 1.243150 \\
22 & 1.047619 & 1.052109 & 1.238095  & 1.222640 \\
  \hline
\end{tabular}
\end{center}
\caption{Summary of lower bounds (LB) and upper bounds on the competitive ratio for the problem, for all even $k$ with $4\leq k\leq 22$.
The lower bounds for the (limited) arrival/departure model are discussed in Section~\ref{sec:departure}. The analysis of $L$-\textsc{Greedy} and AMP carry through to the (limited) arrival/departure model. For $k\geq22$, the upper bound of AMP
is superior to the upper bound of $L$-\textsc{Greedy}. }
\label{table:upper_bounds}
\end{table}
\normalsize

\section{Online matching in the edge arrival/departure model}
\label{sec:departure}

In this section we consider the online matching problem with $k$ edge-recourse
in the setting in which edges may arrive but also {\em depart} online. More
precisely, a request sequence for this problem is of the form
$(p_i,e_i)_{i\geq 1}$, namely the $i$-th request consists of an edge $e_i$ and
its {\em mode} $p_i \in \{\tt{arrive}, \tt{depart}\}$. If $p_i =\tt{arrive}$
then the edge $e_i$ becomes available; this corresponds to the arrival setting
studied in Section~\ref{sec:arrival}. If  $p_i =\tt{depart}$, then $e_i$ is
removed from the graph, and can be used by neither the online algorithm or the
optimal offline  algorithm. We emphasize that in this model, an edge of the
form $(u,v)$ may arrive and depart several times in the course of serving a
request sequence, but every time it arrives it is considered as a ``fresh''
edge. As a consequence, upon each arrival, such an edge is assigned type 0.
Moreover, a departing edge ceases to exist in any matchings.

As explained in the Introduction, we will further distinguish between two
models. In the \emph{limited departure model}, an edge cannot depart  while it
is being used in the matching of the online algorithm,  whereas in the
stronger \emph{full departure model} any edge can depart.

It turns out that the full departure model is quite restrictive. This is
because it is possible for the adversary to force an online algorithm to
augment some augmenting path and then to remove one of the edges in its
matching.  Eventually the algorithm can end up with blocked edges (type $k$),
without having the chance to augment its matching.
This intuition is formalized in the following lemma.
\begin{lemma}
    The competitive ratio of online matching with $k$ edge-recourse in the full departure model is $2$.
    \label{lemma:full}
\end{lemma}

\begin{proof}     
To show an upper bound of 2, consider an algorithm which adds to its matching
any edge whose endpoints are unmatched. The edge types are either $0$ or $1$,
and therefore the algorithm is not sensitive to the given edge budget.  The
matching produced by the algorithm is (inclusion wise) a maximal matching, and
it is well known that its size is at least $1/2$ the size of the maximum
matching.

    To show a lower bound of 2, consider a graph consisting of vertices 1, 2, 3, 4,
    with the edge $(2,3)$ arriving at the beginning.  Then the edges
    $(1,2),(3,4)$ arrive and depart repeatedly, alternating between two
    configurations. When the graph consists of the single edge $(2,3)$, the
    algorithm needs to include it in its matching. When the graph consists of
    the path $(1,2,3,4)$, the algorithm needs to apply this augmenting path if
    it claims to have a competitive ratio strictly lower than $2$.  As a result, the type
    of the edge $(2,3)$ keeps increasing, and when it reaches $k$,
    the algorithm cannot augment the matching anymore.  Thus, for
    even $k$, the algorithm has a matching of size $0$, while the optimal
    matching consists of the edge $(2,3)$. Similarly, for odd $k$, the
    algorithm has a matching of size $1$, while the optimal matching has size
    $2$.  Hence, no algorithm can achieve a competitive ratio strictly lower
    than~$2$.
\end{proof}

Since the full departure model is very restrictive for the algorithm, as shown in Lemma~\ref{lemma:full},
we will concentrate on the limited departure model, as defined in the introduction.  For this model, we observe that the algorithms $L$-\textsc{Greedy} and AMP have the same performance guarantee as in the edge arrival model. This is because the analysis of $L$-\textsc{Greedy} uses weights on vertices which are not affected by edge departures, and the analysis of AMP is based on an upper bound over the number of type $k$ edges, which still holds under edge departures.  We thus focus on obtaining stronger lower bounds in this model (also included in
Table~\ref{table:upper_bounds}). We begin by observing that the bound of $3/2$ of the competitive ratio in the edge arrival model for $k\in\{2,3\}$ still holds for the limited departure model, in which the adversary is stronger.
Hence, the smallest interesting value for $k$ in this model is $k=4$, for which we provide the following specific lower bound. The proof will provide some intuition about the adversarial argument for general $k$, which is shown in 
Theorem~\ref{thm:k.geq.6}.
\begin{theorem}
    The competitive ratio of online matching with $k$ edge-recourse in the \emph{limited departure model} is at least $10/7$ for $k=4$.
\label{thm:departure.k_equal_4}
\end{theorem}
\begin{proof}
We will prove the theorem by applying a game between the online algorithm and the adversary. 
In particular, the adversary will enforce arrivals and departures of edges in such a way that, at every moment in time, 
the symmetric difference between the matching produced by the algorithm 
and the optimal matching consists only of augmenting paths.
In particular, this symmetric difference will have no alternating cycles or alternating paths of even length.

Any such augmenting path can be represented as a {\em string} of integers in $\{0,1,\ldots,k\}$, which is precisely the type of this path.
Note that for an augmenting path, this string begins and ends with 0. We can thus think of the above-defined symmetric difference as a collection
of strings, which in turn allows us to define the game between the algorithm and the adversary over this collection of strings as opposed to defining it
over the actual graph.

Whenever the algorithm applies an augmenting path, this translates into the increment of all edge types of the corresponding string, for example 
the string $01210$ becomes $12321$.  The adversary will respond to this augmentation by a combination of the following three possible types of {\em operations}. 
\begin{itemize}
\item 
The adversary may {\em append} $0$'s to both ends of a string, for example $101 \rightarrow 01010$. 
To do so, the adversary releases two new edges (of type 0). Each edge is incident with only one endpoint of the path described by the string, and is not
incident with any other edge in the current graph.

\item
The adversary may {\em split} the string into smaller strings, for example $12321\rightarrow \{123,1\}$ or $12321\rightarrow \{1,3,1\}$. 
This can be done via the departure of certain edges which are not in the algorithm's matching (of even type). 
For instance, the operation $12321\rightarrow \{123,1\}$ can be done by having one edge of type 2 depart, and the operation
$12321\rightarrow \{1,3,1\}$ can be done by having  both edges of type 2 depart.

\item 
The adversary may {\em merge} certain strings, for example $\{1,1\} \rightarrow 101$.  
This can be done via the arrival of a new edge (of type 0).
\end{itemize}

\begin{figure}[ht]
\centerline{\includegraphics[width=\textwidth]{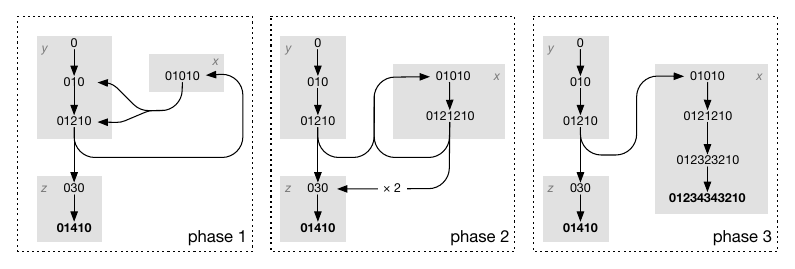}}
\caption{An illustration of the three phases in the game between the algorithm and the adversary, for the proof of 
Theorem~\ref{thm:departure.k_equal_4}.  Blocked strings are depicted in bold face. The arcs illustrate the actions of the adversary, 
after an augmentation by the algorithm. For example, if the algorithm augments the string $01010$ in phase 1, then  
the adversary replaces the resulting string by the strings $010$ and $01210$, whereas in phases 2 and 3 it replaces it with the string $0121210$.  
The numbers $x,y,z$ count the number of strings which belong to the corresponding shown boxes.}
\label{fig:lb_arr_dep_k4}
\end{figure}

\paragraph{The game between the algorithm and the adversary} The main idea
behind the adversarial construction is as follows. We suppose that the online
algorithm has competitive ratio at most  $(10 - \epsilon)/7$ for arbitrarily
small $\epsilon > 0$. We will then show that the adversary can eventually
force the algorithm  to a competitive ratio at least $10/7$, thus leading to a
contradiction.

The game begins with the adversary presenting the string $0$, which the
algorithm has to augment, resulting in a single string $010$.  From this point
onwards, the game proceeds in three {\em phases}, which are depicted in
Figure~\ref{fig:lb_arr_dep_k4}.  In each phase, a sequence of
algorithm/adversary {\em actions} takes place. Each action is of the following
form: The algorithm chooses some string $s$ to augment, which results in a
string $s'$. Then the adversary will perform a sequence of the  above defined
operations on $s'$, which will result in either a single new string (say
$\overline{s}$), or to several new strings, say $\overline{s}_1,
\overline{s}_2, \ldots \overline{s}_k$ (in our construction, it will be that
$k \in [1,3]$). This adversarial action is depicted by means of an arrow from
$s$ to each of the $\overline{s}_1, \overline{s}_2, \ldots \overline{s}_k$ in
Figure~\ref{fig:lb_arr_dep_k4}.

As an example, in phase 1, if the algorithm augments string $s=010$ (thus
obtaining string $s'=121$), the adversary appends two zeros at both ends of
$s'$, which results in the string $\overline{s}_1=01210$. If the algorithm
augments string $s=01210$, thus obtaining string $s'=12321$, then the
adversary first splits $s'$ to two strings $1,3$ and $1$, which he then
transforms into the strings $\overline{s}_1=01010$ and $\overline{s}_2=030$,
by merging them and appending zeros. If the algorithm augments string
$s=030$, thus  obtaining string $s'=141$, then  the adversary appends two
zeros at both ends of $s'$, which results in the string
$\overline{s}_1=01410$. Last, if the algorithm augments string  $s=01010$
(thus obtaining $s'=12121$) then the adversary first splits $s'$ to two
strings $1$ and $121$,  then appends two zeros to the end of each string.
This results in two strings  $\overline{s}_1=010$ and $\overline{s}_2=01210$.
The above are all possible actions that can occur in phase 1.  Actions for
phases 2 and 3 are similar and defined by the graphs in 
Figure~\ref{fig:lb_arr_dep_k4}. 

We also need to explain how the game {\em transitions} between phases; i.e.,
under which conditions the game moves from phase $i$ to phase $i+1$. To this
end, we define some variables which count the numbers of some specific
strings. More precisely:
\begin{itemize}
\item $x$ denotes the number of strings  $01010$, $0121210$, $012323210$ or $01234343210$; such strings have local ratio $3/2$, $4/3$, $5/4$ and $6/5$, respectively. 
\item $y$ denotes the number of strings $0$, $010$ or $01210$; such strings have local ratio $\infty$, $2$ and $3/2$, respectively.
\item $z$ denotes the number of strings $030$ or $01410$; such strings have local ratio $2$ and $3/2$, respectively.
\end{itemize}
In Figure~\ref{fig:lb_arr_dep_k4} we use ``boxes'' to illustrate this grouping of strings. 

In particular, we will call strings $030$ or $01410$ \emph{bad strings}.  
This is motivated by the observation that $01410$ is {\em blocked} and has
large local ratio equal to $3/2$; note that the algorithm cannot augment  such
a string. Moreover, the string $030$ has local ratio 2, and immediately after
an augmentation it becomes $01410$. 

The goal of the adversary is to reach a configuration with only blocked
strings $01410$ and $01234343210$ with a maximum proportion of bad strings
$01410$, since this maximizes the competitive ratio.  It is relatively easy
for the adversary to generate bad strings, but this comes at the expense of
generating strings $01010$.  The adversary's goal is to minimize the
proportion of these strings, and this is done through three different phases. 
 At a high level, the objective of  phase 1 is to create a large number of bad
strings. This is is also the objective of phase 2, but with a more
\emph{efficient} generation of bad strings, in the sense that fewer $01010$
strings are generated per bad string. The objective of phase 3 is simply to
bring the game in a configuration with only blocked strings.

The game starts with the adversary entering phase 1 and generating the single
string $0$. 
In phase 1, immediately after each action of the adversary, the algorithm has
competitive ratio at least $3/2$, meaning that it is forced to augment strings, since $3/2 > (10 - \epsilon)/7$, which is the competitive ratio claimed by the algorithm. This is because in this phase all strings
have local ratio at least $3/2$. 
Throughout phase 1 we have the invariant
\begin{equation}
  2x + y = z + 1, \tag{Inv 1} \label{inv:inv1}
\end{equation}
which can be  verified by inspecting each possible action of phase 1. For example the augmentation of $01010$ decrements $x$ and increments $y$ by $2$.  Eventually the inequality $7z+3 > 2/\epsilon$ will come to hold, simply because after at most $x+1$ augmentations, the counter $z$ increases strictly. At that moment, the adversary moves to phase 2.

Throughout phase 2 we have the invariants
\begin{equation}
7z+3 > 2/\epsilon \tag{Inv 2.1} \label{inv:inv21}
\end{equation}
and
\begin{equation}
2x + y \leq z + 1 \tag{Inv 2.2} \label{inv:inv22}.
\end{equation}

\eqref{inv:inv21} holds because the left-hand side will not decrease
throughout the phase. To show~\eqref{inv:inv22}, we first observe that by
invariant~\eqref{inv:inv1} phase 2 starts with equality, and the actions of
phase 2 preserve the inequality~\eqref{inv:inv22}, which can be easily shown
by inspecting each possible action during phase 2. 

Phase 2 continues for as long as $z<8(x+y)$, and phase 3 begins at the point
in which $z \geq 8(x+y)$. This condition will eventually be reached, because
the quantity $x+y$ is invariant during phase 2, whereas any sequence of at
least $x+1$ actions increases $z$ by at least 1. 

We will now argue that in phase 2, right after each action of the adversary,
the competitive ratio of the algorithm is strictly greater than
$(10-\epsilon)/7$, which implies that the algorithm must, in turn, respond
with an augmentation to every action of the adversary in phase 2, since we
assumed that the algorithm is $(10-\epsilon)/7$-competitive. To this end, we
observe that at each point in phase 2, the  algorithm maintains certain types
of strings whose local ratio we lower bounded above. In particular, there are
$y+z$ strings of local ratio at least $3/2$ and $x$ strings of local ratio at least
$4/3$. Therefore, a
lower bound to the competitive ratio during phase 1, 
can be stated as follows, where we will make use of
the property 
\begin{align*}
f(x,y) := \frac{ax+cy}{bx+dy} \quad  & \textrm{is decreasing on $x$ and increasing on $y$ if} \\
& a,b,c,d > 0 \ \textrm{and}  \ \frac{a}{b} \leq \frac{c}{d}.
\tag{P1}
\end{align*}
We have
\begin{align*}
      \frac{4x + 3y + 3z }
            {3x + 2y + 2z}
    &= 
      \frac{8x+6y+ 6z}
            {6x + 4y +4z} \\
    &=\frac{4(2x+y) + 2y + 6z}{3(2x+y) + y + 4z}\\
    &\geq \frac{4(z+1) + 2y + 6z}{3(z+1)+y+4z} \tag{from~P1 and~\ref{inv:inv22}}\\
    &=\frac{2y+10z+4}{y+7z+3}\\
    &\geq \frac{10z+4}{7z+3} \tag{from~P1}\\ 
    &=\frac{10}{7} - \frac{\frac27}{7z+3}\\
    &> \frac{10}{7} - \frac{\frac27}{\frac{2}\epsilon}  \tag{from~\ref{inv:inv1}} \\
    & = \frac{10-\epsilon}{7}.
\end{align*}

Recall that when the condition $z \geq 8(x+y)$ becomes satisfied, the game moves to the final phase, namely phase 3. 
Moreover, the condition $z \geq 8(x+y)$ holds throughout phase 3, since $x+y$  is invariant and $z$ can only increase in this phase. 
We also obtain that 
\[
y+z\geq y + 8(x+y) \geq 8x,
\]
which will be useful.
Similar to the previous argument, we can lower bound the competitive ratio after each action of the adversary by
\begin{align*}
    \frac{3(y+z) + 6x }{2(y+z) + 5x} &\geq \frac{3\cdot (8x)+6x}{2\cdot(8x)+5x} \tag{From $y+z\geq 8x$ and~P1} \\
    &=\frac{10}{7}.
\end{align*}
Therefore, the algorithm must augment after each action of the adversary, and eventually must find itself in a
configuration which consists only of blocked strings (either $01410$ or $01234343210$). At this configuration, the 
algorithm cannot do any further augmentations, hence its competitive ratio is at least $10/7$, a contradiction. 
\end{proof}

We can generalize the ideas in the proof of Theorem~\ref{thm:departure.k_equal_4} so as to obtain a non-trivial lower bound for general even $k \geq 4$ in Theorem~\ref{thm:k.geq.6}. Note that  since $ \frac{k^2-3k+6}{k^2-4k+7}> 1+\frac{1}{k-1}$ for all $k \geq 4$, Theorem~\ref{thm:k.geq.6} shows a stronger lower bound for even $k$
than Theorem~\ref{thm:lower.bound.deterministic} under the limited departure model.

\begin{theorem}
    The competitive ratio online matching with $k$ edge-recourse in the limited departure model is at least
    $\frac{k^2-3k+6}{k^2-4k+7}$, for all even $k\geq 4$.
\label{thm:k.geq.6}
\end{theorem}

\begin{proof}
First, we observe that for $k = 4$, the expression $\frac{k^2-3k+6}{k^2-4k+7}$ is equal to $10/7$, which is precisely the value obtained in Theorem~\ref{thm:departure.k_equal_4}. Therefore, it suffices to prove the result for even $k \geq 6$.

The proof generalizes the ideas behind the proof of Theorem~\ref{thm:departure.k_equal_4}, and in particular the concept of a game between the online algorithm and the adversary. Again, we suppose that the online algorithm has competitive ratio at most $\frac{k^2-3k+6 - \epsilon}{k^2-4k+7}$ for arbitrarily small 
$\epsilon > 0$. We will then show that the adversary can eventually force the algorithm to a competitive ratio at least $\frac{k^2-3k+6}{k^2-4k+7}$, thus leading to a contradiction.

The game begins with the adversary presenting the string $0$, which the algorithm has to
augment, resulting in a single string $010$. From this point onwards the game proceeds in two phases, which are depicted in 
Figure~\ref{fig:lb_arr_dep_k_general}. In each phase, a sequence of algorithm/adversary actions takes place. Actions are defined to be consistent with Figure~\ref{fig:lb_arr_dep_k_general} for the two phases of the game. 
It is worth pointing out that the game for $k \geq 6$ consists of two phases, while in contrast, the game for $k=4$ (as described in the proof of Theorem~\ref{thm:departure.k_equal_4}) consists of three phases. The second phase for $k=4$ is necessary for the adversary to force the algorithm to keep augmenting strings.


\begin{figure}[ht]
\centerline{\includegraphics[width=12cm]{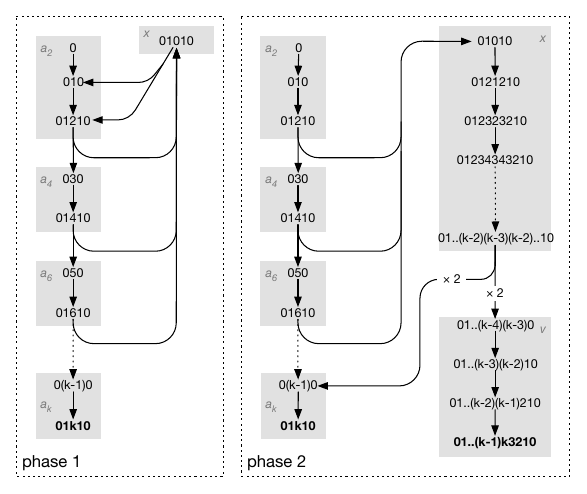}}
\caption{The lower bound construction in the arrival/departure model. Blocked strings are depicted in bold face. The arcs illustrate the adversarial strategy. For example if the algorithm augments the string $01\ldots (k-2)(k-3)(k-2)\ldots 10$ the adversary replaces the resulting string by two strings $0(k-1)0$ and two strings $01\ldots (k-4)(k-3)0$.  The numbers $a_2,a_4,a_6,\ldots,a_k,x,v$ count the number of strings which belong to the corresponding shown boxes.}
\label{fig:lb_arr_dep_k_general}
\end{figure}

We also need to explain under which conditions the game transitions from phase $1$ to phase $2$. To this end, we define again some variables which count the numbers of certain specific strings. More precisely:
\begin{itemize}
\item $x$ denotes the number of strings $0123\ldots i(i-1)i\ldots3210$ for all $1\leq i \leq k-2$; such strings have local ratio $\frac{i+2}{i+1}$. 
\item $a_i$, where $i$ is even in $[2,k]$, denotes the number of strings $0(i-1)0$ or $01i10$; such strings have local ratio $2$ and $3/2$, respectively.
\item $v$ denotes the number of strings $0123\ldots (k-3)0$, $0123\ldots (k-2)10$, $0123\ldots (k-1)210$ or $0123\ldots k 3210$; such strings have local ratio $\frac{k/2}{k/2-1}$, $\frac{k/2+1}{k/2}$, $\frac{k/2+2}{k/2+1}$ and $\frac{k/2+3}{k/2+2}$, respectively. 
\end{itemize}

In phase 1, immediately after each action of the adversary, the algorithm has competitive ratio at least $3/2$. This is because in this 
phase all strings have local ratio at least $3/2$. However, after each augmentation by the algorithm the competitive ratio can be 
much better, and possibly smaller than $\frac{k^2-3k+6}{k^2-4k+7}$. For this reason, the adversary will move eventually the game to phase 2.
In particular, phase 2 begins once the following condition is satisfied.
\[
\sum_{j=1}^{k/2}{(2j(k-1) - 3 k + 7) \cdot a_{2j}} > \frac{2(k-3)}{\epsilon}-(k-1).
\]
Note that this condition will be satisfied because the coefficient $(2j(k-1) - 3 k + 7)$ is non-negative for $j \geq 2$ and increasing in $j$ for $j\geq 1$, and whenever $a_{2j}$ decreases by one, $a_{2(j+1)}$ increases by one during the execution of phase 1. 
Intuitively, the objective in phase 1 is to create a large number of strings counted by $a_4, a_6, \ldots,a_k$,  whereas in phase 2 the objective is to force the algorithm in a configuration with only blocked strings.

Throughout phase 2, we have the invariants
\begin{align}
&\sum_{j=1}^{k/2}{(2j(k-1) - 3 k + 7) \cdot a_{2j}} > \frac{2(k-3)}{\epsilon}-(k-1), \tag{Inv 3.1} \label{eq:inv1}
\\ 
&2x  \leq 1 + \sum_{j = 1}^{k/2}{(2j-3)a_{2j}} - (k-2)v. \tag{Inv 3.2}\label{eq:inv2}
\end{align}
Invariant \eqref{eq:inv1} holds because the left-hand side will not decrease throughout the phase. To show Invariant~\eqref{eq:inv2}, we first observe that at the end of phase 1, it holds that 
\[
2x \leq 1 + \sum_{j = 1}^{k/2}{(2j-3)a_{2j}};
\]
this can be easily shown by induction on the number of actions during phase 1. Thus, at the beginning of phase 2, Invariant~\eqref{eq:inv2} holds, since $v=0$ at that point. Again, a simple inductive argument on the number of actions throughout phase 2 can show that the invariant is maintained.

We will argue that Invariants \eqref{eq:inv1} and \eqref{eq:inv2} imply that throughout phase 2, the competitive ratio is strictly larger than
$\frac{k^2-3k+6 - \epsilon}{k^2-4k+7}$, and hence throughout phase 2 the algorithm is forced to augment strings, until all strings are blocked.
This is because the competitive ratio in phase 2 can be lower bounded by
\begin{align*}
      & \phantom{\geq} \frac{(\frac{k}{2}+3)v  + 3\sum_{j = 1}^{k/2}{a_{2j}} + kx}
          {(\frac{k}{2}+2)v  + 2\sum_{j = 1}^{k/2}{a_{2j}} + (k-1)x}
       \notag
      \\
      &= \frac{(k+6)v +  6\sum_{j = 1}^{k/2}{a_{2j}} + 2kx}
          {(k+4)v +  4\sum_{j = 1}^{k/2}{a_{2j}} + 2(k-1)x}
       \notag
      \\
      &\geq \frac{(k+6)v +  6\sum_{j = 1}^{k/2}{a_{2j}} + k\left(1 + \sum_{j = 1}^{k/2}{(2j-3)a_{2j}} - (k-2)v \right) }
          {(k+4)v +  4\sum_{j = 1}^{k/2}{a_{2j}} + (k-1)\left(1 + \sum_{j = 1}^{k/2}{(2j-3)a_{2j}} - (k-2)v \right)} 
          \tag{From~P1 and \ref{eq:inv2}}
      \\
      &= \frac{(-k^2+3k+6)v + \sum_{j = 1}^{k/2}{\left(k(2j-3)+6\right)a_{2j}} + k} {(-k^2+4k+2)v + \sum_{j = 1}^{k/2}{\left((k-1)(2j-3)+4\right)a_{2j}} + k-1 }. \\
\end{align*}
To complete the proof, it remains to show that
\begin{equation}
\frac{(-k^2+3k+6)v + \sum_{j = 1}^{k/2}{\left(k(2j-3)+6\right)a_{2j}} + k} {(-k^2+4k+2)v + \sum_{j = 1}^{k/2}{\left((k-1)(2j-3)+4\right)a_{2j}} + k-1 }
      >  \frac{k^2-3k+6-\epsilon}
               {k^2-4k+7}.
\label{eq:final.lower.bound.departure}
\end{equation}
By a simple mathematical manipulation, for~\eqref{eq:final.lower.bound.departure} to hold, it suffices that
\begin{equation}      \label{expr:inv_general_k}
v \cdot C_v  + \sum_{j=1}^{k/2}{a_{2j}\cdot C_{2j}} > 2(k-3)-(k-1)\epsilon,
\end{equation}
where $C_v$ and $C_{2j}$ are defined as
\begin{align*}
C_v &= 3(k^2-7k+10)-(k^2-4k-2)\epsilon  \\
C_{2j} &= 2(k-3)(k-2j) + (2j(k-1) - 3 k + 7)\epsilon.
\end{align*}
We observe that $C_v$ and the first additive term in the expression of $C_{2j}$ (i.e. $2(k-3)(k-2j)$) are non-negative for sufficiently small $\epsilon$, $k\geq 6$ and $1\leq j\leq k/2$. Removing these terms from the left hand side of \eqref{expr:inv_general_k}, it suffices to show that
\begin{equation*}
\sum_{j=1}^{k/2}{(2j(k-1) - 3 k + 7)\epsilon \cdot a_{2j}} > 2(k-3)-(k-1)\epsilon.
\label{eq:same.as.invariant}
\end{equation*}
This inequality is precisely Invariant \eqref{eq:inv1}, which concludes the proof.
\end{proof}


\section{Conclusion}
In this paper we provided improved upper and lower bounds for online maximum matching with $k$ edge-recourse.
More specifically, we analyzed two online algorithms for the edge arrival model, namely AMP and $L$-\textsc{Greedy} which seem to be incomparable:
the former is asymptotically superior, in terms of $k$, but the latter has a better performance analysis for small $k$. It would be interesting
to analyze an algorithm that combines the ingredients of these two algorithms, namely, an algorithm that combines the doubling techniques with
augmenting only along short paths. The difficulty in the analysis of such an algorithm lies in that reasonable charging schemes tend to have ``local'' properties, wheres the doubling algorithm applies a ``global'' criterion which does not easily translate into some structural property that can be useful in analysis.

The problems we consider remain challenging even for $k$ as small as 4, and some gap between the upper and lower bounds remains. Bringing this gap
will probably require new ideas and techniques. To this end, it is worth pointing out that the amortization arguments we used in our analysis may have connections to LP-based algorithms (since the dual of the maximum matching problem is a weighted vertex minimization problem). Thus, it would be very interesting to use a duality-based approach, such as dual fitting, towards the design and analysis of improved algorithms.

\bibliographystyle{plain}
\bibliography{OnlineMatching_journal}

\end{document}